\documentclass[11pt]{article}
\usepackage[letterpaper,margin=1in]{geometry}
\usepackage[dvipsnames]{xcolor}
\definecolor{DarkRed}{rgb}{0.65,0,0}
\usepackage{amsmath,amssymb,amsthm}
\usepackage{graphicx}
\usepackage{algorithm}
\usepackage[colorlinks=true,linkcolor=DarkRed,linktoc=page,citecolor=green!50!black,urlcolor=blue]{hyperref}
\usepackage[capitalize]{cleveref}
\newcommand{\email}[1]{\href{mailto:#1}{#1}}
\numberwithin{equation}{section}
\numberwithin{algorithm}{section}
\crefname{section}{section}{sections}
\crefname{subsection}{subsection}{subsections}
\Crefname{figure}{Figure}{Figures}
\usepackage{amssymb}
\usepackage{tcolorbox}
\usepackage{mdframed}
\usepackage{tikz}
\usetikzlibrary{arrows.meta,backgrounds,calc,decorations.pathreplacing,fit,positioning}

\usepackage{algpseudocode}
\usepackage{nicefrac}
\usepackage{booktabs}
\usepackage{thm-restate}

\theoremstyle{plain}
\newtheorem{theorem}{Theorem}[section]
\newtheorem{lemma}[theorem]{Lemma}

\theoremstyle{definition}
\newtheorem{definition}[theorem]{Definition}
\theoremstyle{remark}
\newtheorem{remark}[theorem]{Remark}
\crefname{remark}{Remark}{Remarks}

\newcommand{\R}{\mathbb{R}}
\newcommand{\MST}{\operatorname{MST}}
\newcommand{\diam}{\operatorname{diam}}

\newcommand{\EE}{\mathbb{E}}

\newcommand{\calE}{\mathcal{E}}
\newcommand{\calF}{\mathcal{F}}

\newcommand{\eps}{\varepsilon}
\newcommand{\sse}{\subseteq}
\newcommand{\nf}{\nicefrac}
\newcommand{\cost}{c}
\newcommand{\costh}{\widehat{\cost}}
\newcommand{\polylog}{\mathrm{polylog}}
\newcommand{\meet}{\mathsf{meet}}

\newcommand{\calM}{\mathcal{M}}
\newcommand{\Dhat}{\smash{\widehat{D}}}
\newcommand{\dhat}{\smash{\widehat{d}}}

\definecolor{Acol}{RGB}{214,40,40}    %
\definecolor{Bcol}{RGB}{33,102,196}   %
\definecolor{Ccol}{RGB}{31,150,84}    %
\definecolor{Dcol}{RGB}{142,68,173}   %
\definecolor{Ecol}{RGB}{0,150,160}    %
\definecolor{Fcol}{RGB}{176,122,42}   %
\definecolor{Tour}{RGB}{13,71,161}    %

\title{Approximating Minimum Spanning Trees in Metrics under Semi-Random Corruptions}
\author{Aditya Anand\thanks{University of Michigan (\email{adanand@umich.edu}).}
\and Anupam Gupta\thanks{New York University and Google Research, USA (\email{anupam.g@nyu.edu}). Supported in part by NSF awards CCF-2422926 and CCF-2608359.}
\and Rajesh Jayaram\thanks{Google Research, USA (\email{rkjayaram@google.com}).}
\and Euiwoong Lee\thanks{University of Michigan (\email{euiwoong@umich.edu}). Supported in part by NSF award CCF-2236669.}
\and Ali Vakilian\thanks{Rutgers University (\email{a.vakilian@rutgers.edu}). Supported in part by NSF CAREER Award CCF-2542925.}}
\date{}
\hypersetup{
  pdftitle={Approximating Minimum Spanning Trees in Metrics under Semi-Random Corruptions},
  pdfauthor={Aditya Anand, Anupam Gupta, Rajesh Jayaram, Euiwoong Lee, Ali Vakilian}
}

\begin{document}

\pagenumbering{roman}
\pagestyle{empty}
\maketitle
\thispagestyle{empty}

\begin{abstract}
We study the problem of finding minimum spanning trees in semi-random corrupted metrics. In this problem proposed by Bateni, Dharangutte, Jayaram, and Wang (COLT' 24), we start with a metric space $\calM = (X, d)$. Then, for each pair of points, the adversary can corrupt its distance independently with probability $p$. The algorithm only gets to observe these possibly-corrupted distances, and has to output a spanning tree whose expected cost (according to the true metric) is multiplicatively close to the cost of the minimum spanning tree in the metric space $\calM$. Prior work showed that losing a factor of $\Omega(\log n)$ is necessary, even if the algorithm may adaptively select $O(n/\log n)$ points and probe the true distances between the selected ones. However, no upper bounds were known.

In this work, we study this problem in three regimes, depending on how the adversary can change distances: i.e., whether it can, for each corrupted pair, only increase, decrease or arbitrarily change the distance values. We show the following:
\begin{enumerate}
\setlength{\itemsep}{0pt}
\setlength{\parsep}{0pt}
\item In the increase-only setting, we give a polynomial-time $O(1)$-approximation algorithm for any constant $0 \leq p < 1$.
\item In the decrease-only setting, we give an $O(\log n)$-approximation algorithm for any constant $0 \leq p < 1$. Additionally, we show an $O(\log^4 n)$-approximation algorithm that runs in quasi-polynomial time $n^{O(\log n)}$, for any constant $p \leq 0.1$. Moreover, we show that a polynomial-time algorithm computing an $n^{1/3-\eps}$-approximation is unlikely to exist.
\item In the general setting, where the adversary can both increase and decrease distances, we show an $O(\log n)$-approximation algorithm for every $0 \leq p < \nicefrac{1}{4}$. Further, we show that no $o(n)$-approximation is possible for $\nicefrac{1}{2} \leq p < 1$.
\end{enumerate}

Our algorithms answer the main open question from Bateni, Dharangutte, Jayaram, and Wang's work; moreover, a lower bound due to these authors tells us that the $O(\log n)$-approximations in the latter two settings are information-theoretically best possible.
\end{abstract}

\clearpage
\tableofcontents
\clearpage
\pagenumbering{arabic}
\pagestyle{plain}

\section{Introduction.}
\label{sec:introduction}

We study the minimum spanning tree (MST) problem when the input metric
is observed through a {\em semi-random corruption} process, in which a
random subset of the pairwise distances is corrupted, and the values
of these corrupted distances are chosen adversarially.\footnote{Note
  that we differ from the classical models of Blum and
  Spencer~\cite{blum1995coloring} and Feige and
  Kilian~\cite{feige2001heuristics}, where the randomness lies in
  instance generation and the adversary is restricted to
  solution-preserving modifications; in our setting, the adversary
  chooses the initial metric, the randomness lies in the set of
  distances to be corrupted, and then the corrupted values themselves
  are fully adversarial and can change the optimal solution. However,
  our framework falls squarely into the broad definition of
  semi-random models from \cite{Feige20}, in which ``input instances
  are generated by a process that combines random components with
  adversarial components.''} More formally, there is an unknown
ground-truth metric $(X, d)$ on $n$ points; each pairwise distance is
corrupted with some fixed constant probability $p$, and every
corrupted entry is then replaced by an adversarially chosen value,
subject to one of three regimes that constrain the {\em direction} of
the corruption: in the {\em decrease-only} model, the value of a
corrupted distance may only be adversarially decreased; in the {\em
  increase-only} model, it may only be adversarially increased; and in
the {\em general (unrestricted)} model, it may be replaced by an
arbitrary adversarially chosen value, with no restriction. Given only
the resulting corrupted distances, an $\alpha$-approximation algorithm
must output a spanning tree whose expected cost {\em in the true
  metric} over the random corruptions is at most $\alpha$ times the
cost of the minimum spanning tree (in the true metric).\footnote{We
  emphasize that the algorithm never learns which entries were
  corrupted, and probing an exact distance is assumed to be
  prohibitively expensive.}

The semi-random corruption model is motivated by the data-rich
machine-learning era, where metric spaces are commonly used to model
the similarities between points in large datasets. For instance, a
powerful technique in modern data analysis is to represent data-points
by \emph{neural
  embeddings}~\cite{mikolov2013efficient,he2016deep,devlin2019bert},
which are vectors $x \in \R^d$; the dissimilarity between data-points
is then measured by the metric Euclidean distance between
them. However, such models are not perfect, and errors in the training
of the embedding model can cause imperfect and unpredictable
similarity scores for some pairs of embeddings.  Another example
involves pretrained similarity models, which may accurately
represent similarities for the vast majority of pairs, yet on the
remainder they may be arbitrarily, even adversarially, wrong. A
further example is when the similarities one has for a given dataset
become out of date due to data drift, causing some fraction of the
similarities to become arbitrarily invalidated. This
corrupted-distance setting is precisely the weak-oracle regime of the
{\em weak-strong oracle model} introduced by Bateni, Dharangutte,
Jayaram, and Wang~\cite{bateni2024metric}, who first proposed the
model and initiated the study of the minimum spanning tree problem
within it.

Bateni et al.~\cite{bateni2024metric} gave the first bounds for MST in this model, and their results already show that the difficulty of the problem depends sharply on the type of corruption allowed. They begin with a {\em restricted corruption} model, in which the corrupted distances, though perturbed, are still guaranteed to form a metric; that is, the observed corrupted distances continue to satisfy the triangle inequality. For this metric-preserving corruption setting, they design an algorithm that returns an $O(\sqrt{\log n})$-approximate MST from the corrupted distances alone, and they show this factor to be essentially optimal: no algorithm can compute an $o(\sqrt{\log n})$-approximate MST from these corrupted distances, even if it may adaptively select $O(n/\sqrt{\log n})$ points and is allowed to read the true (uncorrupted) distances between all the selected points. 

Once the metric-preserving restriction is dropped and corrupted values
may be chosen fully adversarially (which is the general corruption
model in our setting), the problem becomes strictly harder:
\cite{bateni2024metric} prove that no algorithm can compute an
$o(\log n)$-approximate MST\footnote{Theorem 6 of \cite{bateni2024metric} states this for $p = \frac{1}{3}$, but it holds for any constant $0 < p < 1$.}, even if it may adaptively select
$O(n/\log n)$ points and probe the true distances between all selected points. In fact, their lower bound applies even to our decrease-only setting. On the algorithmic
side, however, this general adversarial model has remained wide
open. No non-trivial upper bound is known: no approximation of the
MST, nor even of its cost, has been shown achievable, unless
$\Omega(n^2)$ queries to the true distances are made. Indeed, Bateni et
al.~\cite{bateni2024metric} explicitly pose this as an open question,
and conjectured that an $O(\log n)$-approximation is achievable using
only the corrupted distances.
\begin{center}
  \emph{In the general adversarial corruption model, is an $O(\log n)$-approximation\\
    of the MST achievable using only the corrupted distances?}
\end{center}

\subsection{Our Results.}

In this work we resolve the open question above in the affirmative, giving an $O(\log n)$-approximation for the general corruption model that matches the $\Omega(\log n)$ lower bound of Bateni et al.~\cite{bateni2024metric} and confirms their conjecture.
Our results further show that in the adversarial corruption setting,
the difficulty of the problem is determined by the \emph{direction} in
which the adversary may corrupt distances: i.e., whether the corrupted
distances are all more or all less than the original distances, or if
they are allowed to be corrupted arbitrarily.

\textbf{Increase-only.} We start with the case where the corrupted
values may only be larger than the original distances (the {\em increase-only}
setting): we show this case is strictly easier than the general corruption setting, and admits a very good approximation, in polynomial time.

\begin{restatable}[Increase-only corruption]{theorem}{increaseonlythm}%
  \label{thm:increase-only}%
  There is a poly-time algorithm that, for any constant probability of
  corruption $0 \leq p < 1$, computes a spanning tree whose expected cost
  is at most $O_p(1)$ times the MST cost, in the increase-only
  corruption model.
\end{restatable}

\textbf{Decrease-only.} Next, we consider the case where the corrupted
values may only decrease the distances (the {\em decrease-only}
setting). Given the lower bound of~\cite{bateni2024metric}, we can
only hope to find trees whose expected cost is within a logarithmic
factor of the optimum; the next theorem shows how to achieve this
optimal result.

\begin{restatable}[Decrease-only corruption]{theorem}{decreaseonlythm}%
  \label{thm:decrease-only}%
  There is an algorithm that, for any constant probability of corruption
  $0 \leq p < 1$, computes a spanning tree whose expected cost is
  $O_p(\log n)$ times the MST cost, in the decrease-only corruption
  model.
\end{restatable}

However, the algorithm of \Cref{thm:decrease-only} is not efficient:
can we get efficient algorithms that admit poly-logarithmic approximations?
We make progress on this question by presenting an algorithm that
gives a $O(\log^4 n)$ approximation in $n^{O(\log n)}$ time.

\begin{restatable}[Efficient algorithm for decrease-only]{theorem}{efficientdecrease}%
  \label{thm:efficientdecrease}%
  In the decrease-only model with $p \leq 0.1$, there is an algorithm
  that finds a spanning tree whose expected cost is $O(\log^4 n)$
  times the MST, in $n^{O(\log n)}$ time.\footnote{We do not optimize to get the best possible
    $p$, but certain arguments need $p$ to be sufficiently small.}.
\end{restatable}

We also give some evidence that the quasi-polynomial time is
necessary, and a polynomial-time $n^{1/3-\eps}$-approximation is
unlikely to exist for any $\eps > 0$. It is based on distinguishing $n^{2/3+\eps}$-colorable graphs from $G(n,\nf12)$, which
is proved to be hard for {\em low-degree polynomials}~\cite{kothari2023planted}. See \Cref{thm:hard-decrease} in \Cref{sec:lower} for
exact statements and further discussion.

\textbf{Increase-decrease.}  Finally, we turn to the {\em general
  (unrestricted)} corruption model, where corrupted distances may be
replaced by arbitrary adversarially chosen values, resolving the open
question of~\cite{bateni2024metric}.
\begin{restatable}[General corruption]{theorem}{generalthm}%
  \label{thm:general}%
  There is an algorithm that for any $0 \leq p < \frac{1}{4}$, computes
  an $O_p(\log n)$ approximation for the MST in the general corruption
  model.
\end{restatable}

We end with a lower bound showing that if the corruption probability
is large, i.e., $p \in [\nf12, 1]$, then no algorithm can
compute $o(n)$-approximation.
\begin{restatable}[General corruption hardness]{theorem}{generalhardness}%
\label{thm:hard-general-search}%
Fix $p \in [1/2,1]$. For every $n \geq 2$ and every randomized algorithm $\mathcal A$ that only observes the corrupted distances $\hat{d}$, there is an $n$-point metric space $(X,d)$ and an admissible adversary in the general corruption model such that, if $T = \mathcal A(\dhat)$ is the tree returned by the algorithm, then
\[
    \EE[\cost(T)] \geq \frac{n}{4} \cdot \MST(X).
\]
\end{restatable}

With high probability, the algorithm of \Cref{thm:increase-only} achieves
an $O_p(\log n)$-approximation, while those of
\Cref{thm:decrease-only,thm:general} achieve an
$O_p(\log^2 n)$-approximation; see \Cref{rem:high-probability}.

For the increase-decrease model, we do not yet know efficient
algorithms, and we leave this as an intriguing open problem. Our
results are summarized in \Cref{tab:results}.

We emphasize that all of our algorithms operate on the corrupted
distances alone, with {\em no} queries to true distances. Except for our
quasipolynomial-time algorithm, our algorithms are deterministic, and
their expected-cost guarantees are over the random choice of corrupted
pairs alone. The
$\Omega(\log n)$ lower bounds in the decrease-only and general
corruption models from~\cite{bateni2024metric} are
information-theoretic: they hold against computationally unbounded
algorithms, and indeed even against algorithms that may adaptively
select $\Omega(n)$ pairs and query their true (uncorrupted) distances.
Hence, our (information-theoretic) $O(\log n)$-approximate algorithms of
\Cref{thm:decrease-only,thm:general} are optimal in a strong sense.

\textbf{Paper Outline.} After discussing some more related work in
\Cref{sec:related-work}, we give a detailed technical overview of our
ideas in \Cref{sec:overview}. The formal model definition and some
terminology is introduced in \Cref{sec:preliminaries}. We present the
crucial tree-repair lemma and our result for the increase-only setting
in \Cref{sec:tree-repair}; the decrease-only setting is then discussed
in \Cref{sec:decrease-only}, followed by the general increase-decrease
case in \Cref{sec:general-case}. The efficient algorithm for the
decrease-only setting appears in \Cref{sec:efficient-dec}. We end with
our lower bound result in \Cref{sec:lower}. %

\begin{table}[t]
\centering
\renewcommand{\arraystretch}{1.25}
\begin{tabular}{@{}llll@{}}
\toprule
\textbf{Corruption model} & \textbf{Upper bound} & \textbf{Lower bound} &
\textbf{Runtime}\\
\midrule
Increase-only
  & $O_p(1)$~[Thm.~\ref{thm:increase-only}]
  & --- & Polynomial-time\\
Decrease-only
  & $O_p(\log n)$~[Thm.~\ref{thm:decrease-only}]
  & $\Omega(\log n)$~\cite{bateni2024metric}& Inefficient\\
Decrease-only
  & $O(\log^4 n)$~[Thm.~\ref{thm:efficientdecrease}]
  & $\Omega(\log n)$~\cite{bateni2024metric} & Quasipolynomial-time\\
Decrease-only
  & ---
  & $n^{1/3-\eps}$~[Thm.~\ref{thm:hard-decrease}] & Polynomial-time\\
General, $p < \tfrac{1}{4}$
  & $O_p(\log n)$~[Thm.~\ref{thm:general}]
  & $\Omega(\log n)$~\cite{bateni2024metric}& Inefficient\\
General, $p \geq \tfrac{1}{2}$
  & ---
  & $\Omega(n)$ [Thm.~\ref{thm:hard-general-search}] & Inefficient\\
\addlinespace[2pt]
\bottomrule
\end{tabular}
\caption{Approximability of MST in each corruption model. All upper bounds are achieved using only the corrupted distances, with no queries to true distances.} 
\label{tab:results}
\end{table}

\subsection{Related Work.}
\label{sec:related-work}

The idea of pairing an accurate but expensive similarity oracle with a cheap, unreliable one was introduced in practice by Silwal et al.~\cite{silwal2023kwikbucks}, whose algorithm for budgeted correlation clustering uses a cheap-weak signal to steer queries toward an expensive-strong oracle. Their weak signal is assumed to be only boundedly noisy, whereas the model we study allows a random subset of the pairwise distances to be corrupted arbitrarily, even adversarially. This corruption model was introduced by Bateni et al.~\cite{bateni2024metric}, whose main focus, beyond the MST results discussed above, is metric $k$-clustering. The model has since been taken up in follow-up work: Braverman et al.~\cite{braverman2025relative} study fair clustering in this setting, and further recent work examines clustering under related weak-comparison and noisy-distance oracles~\cite{pradhan2026improved,raychaudhury2026metric}.

Our model also relates to work on estimation and optimization from partially corrupted input, most notably robust statistics and learning from untrusted data, where an adversary corrupts an unknown fraction of otherwise-trustworthy data and the goal is to recover the underlying signal~\cite{diakonikolas2019robust,lai2016agnostic,charikar2017learning}; a difference is that we corrupt entries of a single fixed distance matrix, with no distributional assumption on the true metric. It is also loosely related to learning-augmented algorithms (algorithms with predictions), which exploit untrusted machine-learned advice when it is accurate and remain robust when it is not~\cite{lykouris2021competitive,kumar2018improving,mitzenmacher2020algorithms}, though in our setting the corrupted distances are the only access to the input and their corruption is structured rather than a tunable prediction error.

\section{Technical Overview.}
\label{sec:overview}

In this section, we give an overview of our main algorithmic results. First, we will give an overview of the $O(1)$ approximation algorithms for increase-only model, and the tight $O(\log n)$ approximation algorithms for the decrease-only and general cases. After that, we will describe our efficient $n^{O(\log n)}$ time algorithm for the decrease-only model, which is considerably different and technically much more complicated. 

\subsection[Repairing minimum spanning trees and Increase-Only.]{Repairing minimum spanning trees and Increase-Only (\Cref{sec:tree-repair}).} Our first key observation, which forms a basis for almost all of our algorithms is the \emph{existence} of a low-cost spanning tree on the graph induced by only uncorrupted edges. Let $\MST(X)$ denote the cost of the MST in the true metric $d$ on the input point set $X$. We show that there is a high probability event $F$, such that conditional on $F$, there exists a tree $T$, such that the expected cost of the tree $T$ given $F$, is at most $O(1)\MST(X)$. This result holds for the general model of corruption.

\paragraph{Repair and charging scheme.}
Let $T^*$ be a minimum spanning tree in the metric $(X,d)$. We want to \emph{repair} $T^*$ to obtain such a tree $T$ with only uncorrupted edges. We do this as follows. First, we take an Euler tour of $T^*$, and list the vertices in the order they first appear on this path to obtain a path $P(T^*)$, which we refer to as the \emph{path-shortcut} of $T^*$. For an edge $e = (u,v) \in E(P(T^*))$, write $w(e) = d(u,v)$. Then the cost of $P(T^*)$, $\sum_{e \in P(T^*)} w(e)$ is at most $2\MST(X)$, since an Euler tour visits every edge of $E(T^*)$ twice. 

If every edge of $P(T^*)$ were uncorrupted, the path $P(T^*)$ is our desired tree. Otherwise, pick a corrupted edge $e = (u,v)$ of $P(T^*)$. Order the vertices of the path from left to right, and suppose $u$ appears before $v$. We want to replace this edge by an uncorrupted path. To do this, we simply find a vertex $x$ that is closest to $e$ on the path, either on its left or right, so that both the edges $(u,x)$ and $(v,x)$ are uncorrupted. Such an $x$ must exist with high probability (for any constant $p < 1$). Then by the triangle inequality, we can bound
$d(u,v) \leq d(u,x) + d(v,x)$.

Now let $d_P(a,b)$ denote the distance between any two vertices $a$ and $b$ measured along the path. Then by triangle inequality, we have $d(a,b) \leq d_P(a,b)$. This means we can upper bound
\begin{align*}
d(u,v) \leq d(u,x) + d(v,x) \leq d_P(u,x) + d_P(v,x).
\end{align*}

Thus we can account for the cost $d(u,v)$ by charging an amount $2w(f)$ to every edge $f$ on the path between $u$ and $x$.

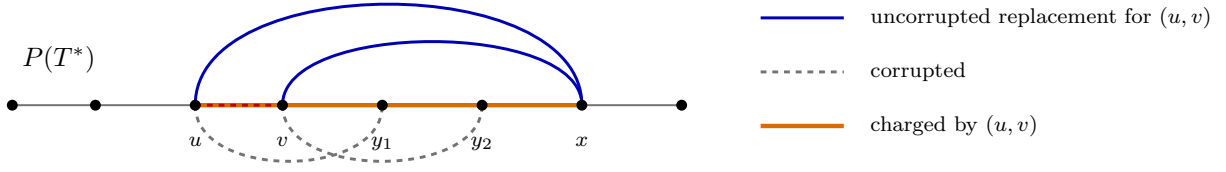
\begin{figure}[t]
\centering
\begin{tikzpicture}[
    x=1.10cm,y=1.10cm,
    vertex/.style={circle,fill=black,inner sep=1.45pt},
    path edge/.style={line width=.75pt,draw=black!55},
    charged path/.style={line width=1.65pt,draw=orange!85!black},
    same dash/.style={dash pattern=on 2.2pt off 2.2pt,line width=1.1pt},
    corrupted edge/.style={same dash,draw=red!75!black},
    aux edge/.style={same dash,draw=black!55},
    good edge/.style={line width=1.12pt,draw=blue!70!black},
    every node/.style={font=\small},
    every label/.style={font=\scriptsize}
]

\coordinate (a0) at (0,0);
\coordinate (a1) at (1,0);
\coordinate (u)  at (2.20,0);
\coordinate (v)  at (3.25,0);
\coordinate (y1) at (4.45,0);
\coordinate (y2) at (5.65,0);
\coordinate (x)  at (6.85,0);
\coordinate (a7) at (8.05,0);

\draw[path edge] (a0)--(a1)--(u)--(v)--(y1)--(y2)--(x)--(a7);

\draw[charged path] (u)--(v)--(y1)--(y2)--(x);

\draw[corrupted edge] (u)--(v);
\draw[aux edge] (u) .. controls +(0,-.9) and +(0,-.9) .. (y1);
\draw[aux edge] (v) .. controls +(0,-0.9) and +(0,-0.9) .. (y2);

\draw[good edge] (u) .. controls +(0,1.62) and +(0,1.62) .. (x);
\draw[good edge] (v) .. controls +(0,1.02) and +(0,1.02) .. (x);

\foreach \P/\L in {a0/{},a1/{},u/$u$,v/$v$,y1/$y_1$,y2/$y_2$,x/$x$,a7/{}}
  \node[vertex,label={[label distance=6pt]below:\L}] at (\P) {};

\node[anchor=west] at (0,0.55) {$P(T^*)$};

\draw[good edge] (9.0,1.05) -- (10.0,1.05);
\node[anchor=west,font=\scriptsize] at (10.2,1.05)
  {uncorrupted replacement for $(u,v)$};

\draw[aux edge] (9.0,0.40) -- (10.0,0.40);
\node[anchor=west,font=\scriptsize] at (10.2,0.40)
  {corrupted};

\draw[charged path] (9.0,-0.25) -- (10.0,-0.25);
\node[anchor=west,font=\scriptsize] at (10.2,-0.25)
  {charged by $(u,v)$};

\end{tikzpicture}

\caption{Repairing a corrupted edge on the path-shortcut. For a corrupted path edge $e=(u,v)$, choose the first meeting vertex $x$ for which both $(u,x)$ and $(v,x)$ are uncorrupted. The replacement path $u-x-v$ is paid for by charging the path interval from $u$ to $x$; each edge $f$ on that interval receives charge at most $2w(f)$.}
\label{fig:tree-repair-charge}
\end{figure}

\paragraph{Analysis.}
Now let us count the expected charge on each edge $f$. If we can show that the charge is at most $O(1) w(f)$, then we would upper bound the cost of $T$ by $O(1)\MST(X)$. Since we only charge an amount $2w(f)$ every time we charge to $f$, it is enough to show that $f$ is charged at most $O(1)$ times. Indeed, this must be the case: an edge $e = (u,v)$ at distance $y$ from $f$ charges $f$ only when it cannot charge any of the closer $y - 1$ edges. This happens with probability at most $(1 - (1-p)^2)^{y}$, since for a vertex $x$, at least one of the edges $(u,x)$ and $(v,x)$ is corrupted with probability exactly $(1 - (1-p)^2)$. Adding across all integers $y \geq 1$, this leads to a geometric sum 
$$\sum_{y \geq 1} (1 - (1-p)^2)^y$$
which is $O(1)$ for every $0 \leq p < 1$. Thus $T$ has expected cost at most $O(1) \MST(X)$.

\paragraph{Algorithm for Increase-Only (\Cref{sec:increase-only}).} Given this, the algorithm for the increase-only model is very simple: just output any minimum spanning tree $T'$ on the graph of corrupted distances $\hat{d}$. With high probability, from the previous paragraph, we know that there is a tree whose edges are uncorrupted and has expected cost at most $O(1)\MST(X)$ (conditional on the high probability event). This tree has no corrupted edges, and $T'$ is a minimum spanning tree with respect to the corrupted distances. Thus $T'$ has expected cost at most $O(1) \MST(X)$ according to $\hat{d}$. But this means that $T'$ has expected true cost at most $O(1)\MST(X)$ according to $d$ as corruption cannot decrease distances. When the high probability event does not hold, we can trivially upper bound the cost of our tree by $n \cdot\MST(X)$, since each edge is of length at most $\MST(X)$.

\subsection[Algorithm for Decrease-Only.]{Algorithm for Decrease-Only (\Cref{sec:decrease-only}).} For the decrease-only setting, first observe that simply finding a minimum spanning tree on the corrupted distances may fail badly. 
\paragraph{An example and our algorithm.}
A simple example is two clusters of points, each of size $\frac{n}{2}$. Inside the cluster, the pairwise distances are all $0$. Across the clusters, every distance is $1$. Clearly, the true minimum spanning tree has cost $1$. But now corruptions with probability $p$ can force a $p$-fraction of the inter-cluster distances to be $0$. If we simply find a minimum spanning tree on the corrupted distances, then this tree, with good probability, can use many inter-cluster edges, leading to a cost of $\Omega(n)$.  

Thus we need a new object and a stronger certificate for spanning trees. It turns out that the right object is that of a \emph{dominating tree}, and this is our second key technical idea. We say that a tree $T$ is dominating, if for every edge $(u,v)$, the corrupted distance $\hat{d}(u,v)$ is at most the corrupted distance along the tree $T$, $\hat{d}_T(u,v)$. A \emph{minimum dominating tree} is a dominating tree whose cost is minimum with respect to the corrupted distances $\hat{d}$. Observe that in the two-cluster example, it removes the key obstruction. For instance, suppose we find a tree whose edges all have distance $0$ after corruption. Now consider some pair $(u,v)$ with $\hat{d}(u,v) = d(u,v) = 1$. Such a pair exists with high probability. But the distance between this pair along the tree $\hat{d}_T(u,v) = 0$, so this tree is not dominating. Observe that there exists a tree on only uncorrupted edges whose true cost is $1$. By the metric property of $d$, any tree on purely uncorrupted edges must be dominating, so the minimum dominating tree has cost $1$, and therefore uses exactly one inter-cluster edge of cost $1$. This, in fact, means that we obtain a partition into two $0$-connected components with respect to $\hat{d}$.

While these two components may not be exactly equal to the two true hidden clusters, they are quite close: formally, we claim that each component has at most $O(\log n)$ points from one of the two true clusters. Suppose there is a component that has at least $C\log n$ points from both true clusters for some large constant $C$. Then since for any pair of points, one from each cluster, that lie in a single component, the corresponding edge must be corrupted, we obtain a $(C\log n, C\log n)$ bi-clique in a random graph, which cannot happen with high probability.

Thus at most $O(\log n)$ points of each cluster are misclassified. If the tree has maximum degree $O(1)$, this means that at most $O(\log n)$ edges have true cost $1$, and hence we obtain a tree whose true cost is at most $O(\log n)$. It is easy to ensure low maximum degree: instead of directly returning a dominating tree, we return its path-shortcut.

\paragraph{Charging scheme.}
Now we show that the minimum dominating tree approach extends to any metric. For simplicity, in this overview, let us assume that the minimum dominating tree $T'$ is a path (if not, we can compute a path-shortcut as discussed above). As we proved earlier, with high probability, there exists a tree using only uncorrupted edges whose expected cost is at most $O(1) \MST(X)$. Since any tree with all uncorrupted edges is dominating, and $T'$ is a minimum dominating tree, the expected cost of $T'$ under $\hat{d}$ must also be $O(1) \MST(X)$. 

We will now show, conditioned on a high probability event, that the expected true cost of $T'$ is at most $O(\log n) \MST(X)$.

We do this via a charging argument similar to the one we previously used for tree-repair, though the analysis and the approximation guarantee will be different. For each edge $f \in E(T')$, let $w(f)$ denote the distance under $\hat{d}$, so that $\sum_{f \in E(T')} w(f)$ is at most $O(1) \MST(X)$ in expectation. Now, we will upper bound the true cost of $T'$ by charging to these weights $w(f)$.

We will upper bound the true cost of $T'$ edge by edge. Order the path edges from left to right. Fix an edge $e = (u,v) \in E(T')$, and suppose $u$ appears first from the left. If $e$ is uncorrupted, the true distance $d(u,v) = \hat{d}(u,v) = w(e)$, so we simply charge to $w(e)$.

Now suppose $e$ is corrupted. Then, we find the first vertex $x$ to the right of $v$ such that both $(u,x)$ and $(v,x)$ are uncorrupted. We can then upper bound $d(u,v) \leq d(u,x) + d(v,x)$. Again, $d(u,x) = \hat{d}(u,x) \leq \hat{d}_{T'}(u,x)$, and $d(v,x) = \hat{d}(v,x) \leq \hat{d}_{T'}(v,x)$ since $T'$ is dominating. Thus we can charge an amount $2w(f)$ to every edge $f$ on the $(u,x)$ path to account for $d(u,x) + d(v,x)$. If no such $x$ exists, let us simply bound the cost by the largest true distance between two points, which is at most $\MST(X)$. In  this case, let us simply say that we charge the value $\MST(X)$.

\paragraph{Analysis.}
Let us now track how many times an edge $f$/the value $\MST(X)$ is charged. If we prove that each edge $f$/$\MST(X)$ is charged at most $O(\log n)$ times, then that would prove that the true cost of $T'$ is at most $O(\log n)\sum_{f \in E(T')} w(f) \leq O(\log n)\MST(X)$, since every time we charge $f$, we charge at most $2w(f)$. Observe that only edges $e$ to the left of $f$ can charge $f$. Next, if two vertex disjoint edges $e_1 = (u_1, v_1)$ and $e_2 = (u_2, v_2)$ charge $f$, then note that one of the edges $(u_1, u_2)$ and $(v_1, u_2)$ has to be corrupted. Similarly, one of the edges $(u_1, v_2)$ and $(v_1, v_2)$ has to be corrupted.

In other words, if we pick a set of vertex disjoint edges $e_1, e_2 \ldots e_{t}$ that charge $f$, then the graph of corrupted edges between them is very dense (almost a clique, in the sense that the above property holds for every pair of these edges). In a random graph with corruption probability $p$, we expect the size of these cliques to be $O_p(\log n)$ with high probability. Thus we must have $t = O(\log n)$, and hence the number of edges that charge $f$ is $O(\log n)$. Following exactly the same proof, we see that the value $\MST(X)$ is charged at most $O(\log n)$ times, finishing the proof. 

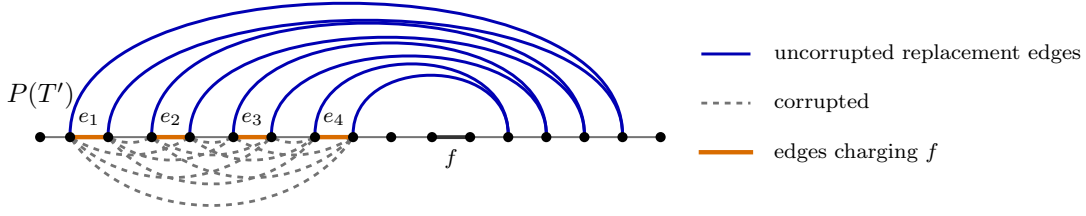
\begin{figure}[t]
\centering
\begin{tikzpicture}[
    x=.72cm,y=1.00cm,
    vertex/.style={circle,fill=black,inner sep=1.35pt},
    path edge/.style={line width=.75pt,draw=black!55},
    charging edge/.style={line width=1.65pt,draw=orange!85!black},
    target edge/.style={line width=1.45pt,draw=black!80},
    same dash/.style={dash pattern=on 2.2pt off 2.2pt,line width=1.05pt},
    corrupted edge/.style={same dash,draw=black!55,shorten <=1pt,shorten >=1pt},
    good edge/.style={line width=1.12pt,draw=blue!70!black,shorten <=1pt,shorten >=1pt},
    every node/.style={font=\small},
    every label/.style={font=\scriptsize}
]

\coordinate (a0) at (-.55,0);

\coordinate (u1) at (0.00,0);
\coordinate (v1) at (0.70,0);

\coordinate (u2) at (1.50,0);
\coordinate (v2) at (2.20,0);

\coordinate (u3) at (3.00,0);
\coordinate (v3) at (3.70,0);

\coordinate (u4) at (4.50,0);
\coordinate (v4) at (5.20,0);

\coordinate (z)  at (5.90,0);
\coordinate (fl) at (6.65,0);
\coordinate (fr) at (7.35,0);

\coordinate (x4) at (8.05,0);
\coordinate (x3) at (8.75,0);
\coordinate (x2) at (9.45,0);
\coordinate (x1) at (10.15,0);
\coordinate (aR) at (10.85,0);

\draw[path edge]
  (a0)--(u1)--(v1)--(u2)--(v2)--(u3)--(v3)--(u4)--(v4)--(z)--(fl)--(fr)--(x4)--(x3)--(x2)--(x1)--(aR);

\draw[corrupted edge] (v1) to[bend right=24] (u2);
\draw[corrupted edge] (u1) to[bend right=30] (v2);

\draw[corrupted edge] (v2) to[bend right=24] (u3);
\draw[corrupted edge] (u2) to[bend right=30] (v3);

\draw[corrupted edge] (v3) to[bend right=24] (u4);
\draw[corrupted edge] (u3) to[bend right=30] (v4);

\draw[corrupted edge] (v1) to[bend right=38] (u3);
\draw[corrupted edge] (u1) to[bend right=43] (v3);

\draw[corrupted edge] (v2) to[bend right=38] (u4);
\draw[corrupted edge] (u2) to[bend right=43] (v4);

\draw[corrupted edge] (v1) to[bend right=50] (u4);
\draw[corrupted edge] (u1) to[bend right=55] (v4);

\draw[good edge] (u1) .. controls +(0,2.35) and +(0,2.35) .. (x1);
\draw[good edge] (v1) .. controls +(0,2.05) and +(0,2.05) .. (x1);

\draw[good edge] (u2) .. controls +(0,2.00) and +(0,2.00) .. (x2);
\draw[good edge] (v2) .. controls +(0,1.75) and +(0,1.75) .. (x2);

\draw[good edge] (u3) .. controls +(0,1.65) and +(0,1.65) .. (x3);
\draw[good edge] (v3) .. controls +(0,1.42) and +(0,1.42) .. (x3);

\draw[good edge] (u4) .. controls +(0,1.28) and +(0,1.28) .. (x4);
\draw[good edge] (v4) .. controls +(0,1.08) and +(0,1.08) .. (x4);

\draw[charging edge] (u1)--(v1) node[midway,yshift=7pt,font=\scriptsize] {$e_1$};
\draw[charging edge] (u2)--(v2) node[midway,yshift=7pt,font=\scriptsize] {$e_2$};
\draw[charging edge] (u3)--(v3) node[midway,yshift=7pt,font=\scriptsize] {$e_3$};
\draw[charging edge] (u4)--(v4) node[midway,yshift=7pt,font=\scriptsize] {$e_4$};

\draw[target edge] (fl)--(fr) node[midway,yshift=-8pt,font=\scriptsize] {$f$};

\foreach \P in {a0,u1,v1,u2,v2,u3,v3,u4,v4,z,fl,fr,x4,x3,x2,x1,aR}
  \node[vertex] at (\P) {};

\node[anchor=west] at (-1.35,0.55) {$P(T')$};

\draw[good edge] (11.55,1.10) -- (12.55,1.10);
\node[anchor=west,font=\scriptsize] at (12.75,1.10)
  {uncorrupted replacement edges};

\draw[corrupted edge] (11.55,0.45) -- (12.55,0.45);
\node[anchor=west,font=\scriptsize] at (12.75,0.45)
  {corrupted};

\draw[charging edge] (11.55,-0.20) -- (12.55,-0.20);
\node[anchor=west,font=\scriptsize] at (12.75,-0.20)
  {edges charging $f$};

\end{tikzpicture}
\caption{The analysis for the decrease-only algorithm. Here 
$e_1,\ldots,e_4$ all charge the same edge $f$, which implies many corrupted edges between their endpoints.}
\label{fig:decrease-only-charging}
\end{figure}

It may be instructive to compare this $O(\log n)$-approximation with the aforementioned $O(1)$-approximation for {\em repairing minimum spanning trees}, though the two analyses use almost identical charging schemes. The simple but crucial difference between the two is: the $O(1)$-approximation analyzed only one tree $T^*$ (which is the MST of the true metric) chosen before the random corruptions, but the $O(\log n)$-charging bound works with high probability for every dominating tree simultaneously. The $\Omega(\log n)$-lower bound for the decrease-only model~\cite{bateni2024metric} shows that such a conservative analysis is necessary. 

\subsection[Algorithm for General Corruptions.]{Algorithm for General Corruptions (\Cref{sec:general-case}).} One key assumption in the case of decrease-only corruptions was that any tree $T$ with only uncorrupted edges is a dominating tree. However, this assumption is false if the corruptions can increase distances: for some pair $(u,v)$, $\hat{d}(u,v) > d(u,v)$. In this case, $d(u,v) \leq d_T(u,v) = \hat{d}_T(u,v)$ by triangle inequality, but $\hat{d}(u,v)$ may be much larger, violating the domination property. Thus there may not exist any dominating tree with low cost. 

We address this challenge by resorting to \emph{almost-dominating} trees. Here the key observation is the following: if $T$ is a tree only on uncorrupted edges, then for some pair $(u,v)$, if $(u,v)$ is uncorrupted, then $(u,v)$ satisfies the domination condition, since $\hat{d}(u,v) = d(u,v) \leq d_T(u,v) = \hat{d}_T(u,v)$. While we do not know which edges are corrupted, for any set $S$ of size $C\log n$ or larger for some sufficiently large constant $C$, we expect only roughly a $p$ fraction of edges inside $S$ to be corrupted. This leads to the definition of almost-dominating trees: a tree $T$ is almost-dominating if for every subset $S \subseteq X$ of size larger than $C\log n$ for some sufficiently large $C$, at most (roughly) a $p$-fraction of the pairs $(u,v)$ inside $S$ violate the domination property. Note again that this condition can be checked by looking only at $\hat{d}$ since the domination property depends only on $\hat{d}$.

It is now not too hard to see that the previous proof for the decrease-only case now goes through for the general case, once we replace the notion of dominating trees by that of almost-dominating trees. 

\subsection[Efficient quasi-polynomial time for Decrease-Only Corruptions.]{Efficient quasi-polynomial time for Decrease-Only Corruptions (\Cref{sec:efficient-dec}).} One may wonder if given the observed distances, we can compute an approximate minimum dominating tree efficiently. In Appendix \ref{app:dominating-tree-hardness} we show that this is not the case: this problem is graph-coloring hard. Our efficient algorithm is instead based on hierarchical decompositions. In this overview, we only try to motivate the main ideas behind hierarchical decompositions instead of describing the full algorithm: \Cref{sec:efficient-dec} provides an overview of the algorithm itself. 

\paragraph{Clique partition for lower bounding $\MST(X)$.}
First, we define a clique partition, and show how it can be useful to lower bound the cost of an MST. Given some number $t$, let $G_t$ and $\hat{G}_t$ denote the original and corrupted \emph{threshold graphs}, obtained by restricting to distances of value at most $t$. Let $N_t$ be the size of the minimum clique partition in $\hat{G}_t$: the minimum number of disjoint cliques that cover the whole graph $\hat{G}_t$. Then a key lower bound is that the cost of the minimum spanning tree must be at least $\frac{t}{4}(N_t - 1)$.

This can be proved as follows. Consider the true metric, and apply repeatedly the following greedy carving procedure. Pick an uncovered vertex $u$, pick the ball of radius $\frac{t}{2}$ around $u$ in the true metric and mark each point in the ball as covered (if not covered already). Suppose this process gives rise to $x$ many balls. Then the cost of any spanning tree $T$ must be at least $\frac{(x-1)t}{4}$. Why? We take an Euler tour of $T$, and look at the order in which the vertices first appear. In this path, the distance between any two centers of the balls must be at least $\frac{t}{2}$ by construction. Thus the cost of this path, which is at most twice the cost of the tree $T$, must be at least $(x-1)\frac{t}{2}$, giving us the desired bound. 

How does this connect to clique partitions? Since the distance between any two points in a ball is at most $t$, each ball forms a clique in $G_t$. But a clique in $G_t$ is also a clique in $\hat{G}_t$ as the corruptions are decrease-only. Thus if the greedy procedure produced $x$ balls, we obtain a partition of the point set into $x$ cliques of $\hat{G}_t$. Since $N_t$ is the number of cliques in a \emph{minimum clique partition}, we have $N_t \leq x$. Thus the cost of any minimum spanning tree is at least $\frac{t}{4}(N_t - 1)$\footnote{This procedure is similar to the ball-carving introduced in~\cite{bateni2024metric}}.  

\paragraph{Hierarchical clique partition for computing a tree.}

Now suppose that every distance is between $1$ and $\Delta$, and suppose we know the number $\Delta$. We construct a \emph{hierarchical clique partition} (see \Cref{fig:cluster-tree}). On a high level, this is simply a sequence of partitions $\mathcal{C}_0, \mathcal{C}_1, \ldots \mathcal{C}_\ell$ of the point set $X$ at scales $t = 2^i$, $i = 0,1 \ldots \ell = \log \Delta$. Each $\mathcal{C}_{i-1}$ is a refinement of $\mathcal{C}_i$, $\mathcal{C}_0$ is a partition into singletons, and $\mathcal{C}_{\ell}$ is the whole point set. Each $\mathcal{C}_i$ is a partition of the whole point set $X$ into cliques in $\hat{G}_{2^i}$. Similar to the one level minimum clique partition, if we assume that our hierarchical partition is minimum in the sense that it minimizes $\sum_i 2^i |\mathcal{C}_i|$ (where $|\mathcal{C}_i|$ is the number of cliques in the partition of $X$ induced by $\mathcal{C}_i$), we can show via a similar greedy carving procedure that $\sum_i 2^i |\mathcal{C}_i| \leq \MST(X) \log \Delta$.

Now consider each cluster $C \in \mathcal{C}_i$ at some level $i$ of the partition. Suppose we pick one \emph{leader} vertex from each such cluster $C$, and form a hierarchical tree by connecting the leaders in the natural way: the leader of cluster $C \in \mathcal{C}_i$ connects to the leader of the unique cluster $C' \in \mathcal{C}_{i + 1}$ where $C \subseteq C'$ (\Cref{fig:cluster-tree-leader}). 

What is the cost of this hierarchical tree? Suppose for a moment that $C$ and $C'$ were cliques not just in $\hat{G}_{2^{i}}$ and $\hat{G}_{2^{i + 1}}$, but also in $G_{2^i}$ and $G_{2^{i + 1}}$ respectively. Then the distance between the two leaders of $C$ and $C'$ can be bounded by $O(2^i)$. We charge this cost to the child cluster $C$. Thus each parent-child edge of the hierarchical tree charges its child cluster in this manner. The total charge on elements in $\mathcal{C}_i$ would be $O(2^i |\mathcal{C}_i|)$, where $|\mathcal{C}_i|$ is the number of cliques at scale $2^i$. So the total cost of our tree is at most 

$$\sum_{i} 2^i |\mathcal{C}_i| \leq O(\MST(X) \log \Delta).$$

\paragraph{Building a good hierarchical clique partition.}
However, we need to remove multiple assumptions that we made in the above analysis. In particular, first, this scheme is not efficient as it requires finding a minimum cost hierarchical clique partition. Second, we assumed that the clusters $C \in \mathcal{C}_i$ are cliques in $G_{2^i}$, whereas in reality they are only cliques in $\hat{G}_{2^i}$, thus there may be points that are outliers, which are in reality, according to the true distances $d$, very far away from other points of the observed clique. These outliers affect the analysis of the tree cost, as the bounds are no longer correct if we pick an outlier as a leader.

The key idea is that if there are at most $O(\log n)$ outliers in every cluster $C \in \mathcal{C}_i$ for every $i$, then randomly sampling a leader uniformly at random from $C$ is likely to miss the outliers with good probability. We then show via a charging scheme that when we do sample an outlier as a leader, we can bound the costs effectively, again as long as every cluster $C \in \mathcal{C}_i$ has at most $O(\log n)$ outliers. 

To make the algorithm efficient, we employ a greedy carving algorithm similar to the procedure we described earlier. Given a scale $t$, to carve one cluster, we start by finding a small clique of size $100\log n$ in $\hat{G}_t$. This can be done by simply enumerating all sets of size $100\log n$ in quasi-polynomial time. Next, whenever we find a clique in $\hat{G}_t$, it is possible to show that all except at most $20 \log n$ points (outliers) are distance at most $2t$ from each other (see \Cref{lemma:coreoutliereff}), thus this small set of size $100\log n$ is actually close to a true clique. Then, with this set of $100\log n$ points as a reference seed set, we find points that are at distance at most $O(t)$ from at least $80 \log n$ points from this reference seed set, and carve them out. This procedure essentially simulates on $\hat{d}$ the ideal greedy carving procedure we described earlier, which works on the true metric $d$. 

Finally, the fact that our decomposition is hierarchical makes it harder to maintain that the number of outliers is at most $O(\log n)$ in every cluster of every $\mathcal{C}_i$ for all $i$. We build the hierarchy bottom up, so outliers from lower levels can accumulate, leading to a much larger number of outliers as we move up the tree. To resolve this, our actual carving rule also tests candidate points against all possible higher-scale seed sets, enforcing a triangle-inequality-like consistency condition which prevents outliers collected at lower levels from accumulating when clusters merge higher in the hierarchy.

\subsection{Preliminaries.}
\label{sec:preliminaries}

Let $(X, d)$ be a metric space with $n$ points. The diameter of a set
$S \sse X$ is $\diam(S) := \max_{u,v\in S} d(u,v)$. We use the shorthand $\Delta$ for $\diam(X)$. Let $D$ be the $n \times n$ (symmetric)
distance matrix with $D_{i,j} = d(i,j)$ for $1 \leq i, j \leq n$.
\begin{definition}[Semi-Random Corruptions Model]%
  \label{def:corrupt}%
  The distance matrix is \emph{corrupted} semi-randomly as follows:
  \begin{enumerate}
  \item Each unordered pair of indices, given by $\{i,j\}$ with
    $i < j$, is sampled independently with probability $p \in
    [0,1]$. Let $S \subset \binom{n}{2}$ be the set of sampled pairs;
    we call $S$ the \emph{corrupted set}.
  \item An adversary, who knows both $D$ and $S$, can arbitrarily
    choose some value $\dhat(i,j)$ for each pair $\{i,j\} \in S$. Now
    set $\Dhat_{i,j} = \Dhat_{j,i} = \dhat(i,j)$ for all
    $\{i,j\} \in S$, and $\Dhat_{i,j} = \Dhat_{j,i} = D_{i,j}$ for all
    $i,j \not\in S$. This (symmetric) matrix $\Dhat$ is called the
    \emph{corrupted distance matrix}.
  \end{enumerate}
\end{definition}
Note that $\Dhat$ need not satisfy the triangle inequality, and
hence may not be a metric.

\begin{definition}[Uncorrupted/Corrupted Tree Cost]%
  Given a tree $T$, let $\cost(T) := \sum_{\{u,v\} \in E(T)} {d}(u,v)$
  to be the true cost of the tree, and
  $\costh(T) := \sum_{\{u,v\} \in E(T)} \dhat(u,v)$ to be the
  corrupted cost.
\end{definition}

Define $\MST(X)$ to be the (true) \emph{cost} of the minimum spanning
tree of the metric $(X,d)$, i.e., $\MST(X) := \min_{T} \cost(T)$,
where the minimum is taken over all spanning trees of the metric. The
main problem of interest in this paper is the following: 
\begin{definition}[Corrupted MST Problem: $C$-Approximation]%
  An algorithm is a $C$-approximation to the corrupted MST problem if,
  for any corrupted distance matrix $\Dhat$ given as input,
  it outputs a tree $T = (X,E)$ on the metric space $(X,d)$, such that
  \[ \EE[\cost(T)] \leq C \cdot \MST(X). \]
  The expectation above is taken over the randomness of the corruption
  process, as well as any internal randomness of the algorithm itself. 
\end{definition}
We consider two natural special cases of this problem.
\begin{definition}[Increase-Only/Decrease-Only  Model]%
  \label{def:inc-only}%
  The \emph{increase-only} (resp.\ \emph{decrease-only}) corruptions
  model is a restricted version of \Cref{def:corrupt} where the
  adversary is only allowed to increase (resp.\ decrease) the values
  of distances for pairs in the corrupted set $S$. That is, for all
  $\{i,j\} \in S$, we have $\dhat({i,j}) \geq d(i,j)$ (resp.\
  $\dhat({i,j}) \leq d(i,j)$). %
\end{definition}
We sometimes
refer to the general problem (where both increases and decreases are
permitted) as the \emph{increase-decrease} model, when it helps make
the distinction clearer.

We use standard graph-theoretic terminology: for a graph $G = (V,E)$,
we also use $V(G), E(G)$ to denote the vertex and edge sets of $G$. We
often view the metric $(X,d)$ as a complete weighted graph
$(X, \binom{X}{2})$. By the metric property, the cost of any minimum
spanning tree is at least $\Delta$.  The following threshold graph is
a subgraph of this complete graph.

\begin{definition}[Threshold Graphs]%
  Given $t \geq 0$, we denote by $G_{t}$ the graph obtained by
  restricting the true metric $d$ to the pairs $(u,v)$ with distance
  $d(u,v)$ at most $t$. Similarly, we denote by $\widehat{G}_{t}$ the
  graph obtained by restricting the corrupted distances $\dhat$ to
  only those pairs $(u,v)$ whose corrupted distance $\dhat(u,v)$ is at
  most $t$.
\end{definition}

Finally, two conventions: we hide the dependency on the probability of
corruption $p$ inside the big-O notation for our approximation
factors. Moreover, we use the phrase ``with high probability'' to
denote probability at least $1 - \frac{1}{n^{c}}$ for some constant
$c > 0$ where $n$ is the number of points, and again hide the constant
$c$ in the big-O notation. For simplicity, in some of our statements,
we assume that $n$ is large enough.

\section{The Tree-Repair Lemma and the Increase-Only Algorithm.}
\label{sec:tree-repair}

In this section, we show a key tree-repair lemma which works in the
general increase-decrease setting. Given any true minimum spanning
tree $T^*$ for the metric $(X,d)$, we show that there is a high
probability event $\calF$, such that conditioned on $\calF$, one can
obtain a tree $T$ by modifying $T^*$, such that $T$ has expected cost
at most $O(\MST(X))$. As we will see, this immediately leads to an
algorithm for the increase-only case. Recall that, given a tree
$T = (X,E)$, the uncorrupted and corrupted costs are given by
$\cost(T)$ and $\costh(T)$ respectively.  We need the following useful
definition.

\begin{definition}[Path shortcut]%
  \label{def:path-shortcut}%
  Given a tree $T$, the \emph{path-shortcut} $P(T)$ is obtained by
  taking an Euler tour of $T$, and listing vertices in the order in
  which they first appear in the Euler tour.
\end{definition}

\begin{lemma}[Tree Repair Lemma]%
  \label{lemma:repair}%
  Consider the increase-decrease model of corruption with corruption
  probability $p \in [0,1)$. Then there is an event $\calF$ such that
  $\Pr[\calF] \geq 1 - \nicefrac{1}{n^{10}}$, so that (a) conditioned on
  $\calF$, there exists a tree $T$, such that every edge of $T$ is
  uncorrupted, and moreover,
  \begin{gather}
    \EE[\,\cost(T) \mid \calF\,] = \EE[\,\costh(T) \mid \calF\,] \leq O(1)\cdot \MST(X).
  \end{gather}
\end{lemma}

\begin{proof} 
  Let $T^*$ be a minimum spanning tree in the uncorrupted metric, so that $\cost(T^*) = \MST(X)$, and let $P^* = P(T^*)$ denote the path-shortcut of $T^*$ (as defined in \Cref{def:path-shortcut} and illustrated in \Cref{fig:path-shortcut}).

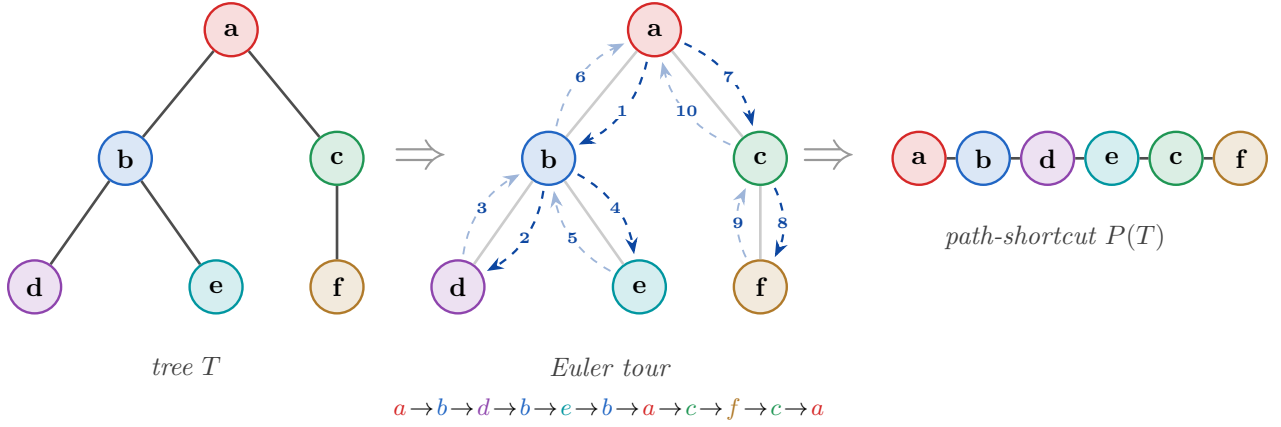
\begin{figure}[t]
\centering
\begin{tikzpicture}[
    >=Stealth,
    A/.style={circle,draw=Acol,fill=Acol!16,line width=0.9pt,minimum size=7mm,font=\small\bfseries},
    B/.style={circle,draw=Bcol,fill=Bcol!16,line width=0.9pt,minimum size=7mm,font=\small\bfseries},
    C/.style={circle,draw=Ccol,fill=Ccol!16,line width=0.9pt,minimum size=7mm,font=\small\bfseries},
    D/.style={circle,draw=Dcol,fill=Dcol!16,line width=0.9pt,minimum size=7mm,font=\small\bfseries},
    E/.style={circle,draw=Ecol,fill=Ecol!16,line width=0.9pt,minimum size=7mm,font=\small\bfseries},
    F/.style={circle,draw=Fcol,fill=Fcol!16,line width=0.9pt,minimum size=7mm,font=\small\bfseries},
    tedge/.style={line width=1pt, black!70},
    ledge/.style={line width=1pt, black!20},
    tourf/.style={Tour, dashed, -{Stealth[length=2.2mm]}, line width=0.85pt,
                  shorten >=2pt, shorten <=2pt},
    tourb/.style={Tour!40, dashed, -{Stealth[length=2.2mm]}, line width=0.7pt,
                  shorten >=2pt, shorten <=2pt},
    num/.style={inner sep=1pt, fill=white, text=Tour, font=\tiny\bfseries},
    panel/.style={font=\small\itshape, black!75}
  ]

  \begin{scope}
    \node[A] (a) at (1.4,3.4)  {a};
    \node[B] (b) at (0,1.7)    {b};
    \node[C] (c) at (2.8,1.7)  {c};
    \node[D] (d) at (-1.2,0)   {d};
    \node[E] (e) at (1.2,0)    {e};
    \node[F] (f) at (2.8,0)    {f};
    \draw[tedge] (a)--(b); \draw[tedge] (a)--(c);
    \draw[tedge] (b)--(d); \draw[tedge] (b)--(e); \draw[tedge] (c)--(f);
    \node[panel] at (0.8,-1.05) {tree $T$};
  \end{scope}

  \node[black!40] at (3.9,1.7) {\huge$\Rightarrow$};

  \begin{scope}[xshift=5.6cm]
    \node[A] (a2) at (1.4,3.4)  {a};
    \node[B] (b2) at (0,1.7)    {b};
    \node[C] (c2) at (2.8,1.7)  {c};
    \node[D] (d2) at (-1.2,0)   {d};
    \node[E] (e2) at (1.2,0)    {e};
    \node[F] (f2) at (2.8,0)    {f};
    \draw[ledge] (a2)--(b2); \draw[ledge] (a2)--(c2);
    \draw[ledge] (b2)--(d2); \draw[ledge] (b2)--(e2); \draw[ledge] (c2)--(f2);

    \draw[tourf] (a2) to[bend left=26] node[num,pos=.5]{1}  (b2);  %
    \draw[tourf] (b2) to[bend left=26] node[num,pos=.5]{2}  (d2);  %
    \draw[tourb] (d2) to[bend left=26] node[num,pos=.5]{3}  (b2);  %
    \draw[tourf] (b2) to[bend left=26] node[num,pos=.5]{4}  (e2);  %
    \draw[tourb] (e2) to[bend left=26] node[num,pos=.5]{5}  (b2);  %
    \draw[tourb] (b2) to[bend left=26] node[num,pos=.5]{6}  (a2);  %
    \draw[tourf] (a2) to[bend left=26] node[num,pos=.5]{7}  (c2);  %
    \draw[tourf] (c2) to[bend left=26] node[num,pos=.5]{8}  (f2);  %
    \draw[tourb] (f2) to[bend left=26] node[num,pos=.5]{9}  (c2);  %
    \draw[tourb] (c2) to[bend left=26] node[num,pos=.5]{10} (a2);  %

    \node[panel] at (0.8,-1.05) {Euler tour};
    \node[font=\footnotesize] at (0.8,-1.62)
      {$\textcolor{Acol}{a}\!\to\!\textcolor{Bcol}{b}\!\to\!\textcolor{Dcol}{d}\!\to\!
        \textcolor{Bcol}{b}\!\to\!\textcolor{Ecol}{e}\!\to\!\textcolor{Bcol}{b}\!\to\!
        \textcolor{Acol}{a}\!\to\!\textcolor{Ccol}{c}\!\to\!\textcolor{Fcol}{f}\!\to\!
        \textcolor{Ccol}{c}\!\to\!\textcolor{Acol}{a}$};
  \end{scope}

  \node[black!40] at (9.3,1.7) {\huge$\Rightarrow$};

  \begin{scope}[xshift=10.5cm, yshift=1.7cm]
    \foreach \i/\lbl/\st in {0/a/A,1/b/B,2/d/D,3/e/E,4/c/C,5/f/F}{
      \node[\st] (p\i) at (0.85*\i,0) {\lbl};
    }
    \draw[tedge] (p0)--(p1)--(p2)--(p3)--(p4)--(p5);
    \node[panel] at (1.8,-1.05) {path-shortcut $P(T)$};
  \end{scope}

\end{tikzpicture}
\caption{The path-shortcut construction of \Cref{def:path-shortcut}. Starting from a tree $T$ (left), an Euler tour visits the vertices in the numbered order $1,\dots,10$ (middle); keeping each vertex at its first occurrence yields the path-shortcut $P(T)$ (right).}
\label{fig:path-shortcut}
\end{figure}

  For each edge $e = \{u,v\} \in E(P^*)$, let $w(e)$ be the sum of
  the true distances of the tree edges corresponding to the sub-path
  between $u$ and $v$ on the Euler tour used to construct
  $P^*$. Notice that $w(P^*) = \sum_e w(e) \leq 2 \cost(T^*)$,
  since an Euler tour traverses every edge of the tree at most
  twice. Given any two vertices $u,v \in V(P^*)$, define $d_P(u,v)$
  to be the total weight of the sub-path between $u$ and $v$. Observe
  that $d(u,v) \leq d_P(u,v)$ by the triangle inequality applied to
  $d(.,.)$.

  Next, we modify the path $P^*$ into a tree $T$. For each corrupted
  edge $e = (u,v)$ on the path where $u$ appears before $v$ when
  scanned left to right, let $\textsf{meet}_r(e)$ be the first vertex
  $x$ on the path to the right of $v$ such that both $(u,x)$ and
  $(v,x)$ are uncorrupted. Similarly, let $\textsf{meet}_l(e)$ be the
  first vertex $x$ on the path to the left of $u$ such that both
  $(u,x)$ and $(v,x)$ are uncorrupted.  The high probability event
  $\calF$ is defined as follows: for every corrupted edge $e$, either
  $\textsf{meet}_r(e)$ or $\textsf{meet}_l(e)$ exists. It is easy to
  see that this holds with high probability as follows. Let
  $r = (1-p)^2$. For every fixed vertex $x \neq u,v$, both the edges
  $(u,x)$ and $(v,x)$ exist with probability $r$. Thus the probability
  no such vertex exists is $(1-r)^{n-2}$. Applying a union bound over
  the $n - 1$ edges $(u,v)$ of $P^*$, this event $\calF$ does not
  happen with probability bigger than $\nicefrac{1}{n^{10}}$ for large
  enough $n$ for every $p < 1$.

  Finally, we obtain $T$ as follows. For each uncorrupted edge of
  $P^*$ keep it as such. For every corrupted edge $e = (u,v)$, if
  $\textsf{meet}_r(e)$ exists, then replace $u,v$ by the path
  $u-\textsf{meet}_r(e)-v$. Else if $\textsf{meet}_l(e)$ exists,
  replace $u,v$ by the path $u-\textsf{meet}_l(e)-v$. Finally,
  arbitrarily delete any additional edges if there are cycles to make
  the resulting graph a tree.

  For the analysis, we first bound the expected cost of $T$ (without
  normalizing by $\Pr[\calF]$ for conditional expectation, which we
  will do at the end). The total cost of uncorrupted edges of $P^*$ is
  clearly $O(\MST(X))$. Every time we replace an edge $e = (u,v)$ by
  the path $u - x - v$, we can bound the cost by charging at most
  twice to the edges on the sub-path on $P^*$ between $u,v$ and
  $x$. More formally, the cost of the path $u-x-v$ is at most
  $d(u,x) + d(x,v) \leq d_{P}(u,x) + d_{P}(x,v)$. We can further
  charge $d_{P}(u,x)$ and $d_{P}(v,x)$ to the weight $w(f)$ of every
  edge $f$ on the unique path between $u$ and $x$, and $v$ and $x$
  respectively in $P^*$. Thus it is enough to show that in this
  charging scheme, an edge is charged at most a constant number of
  times in expectation. If this happens, since the total weight
  $\sum_{f \in P^*} w(f)$ is at most $O(1) \MST(X)$, the total cost of
  $T$ is at most $O(1) \MST(X)$.

  Let us fix an edge $f$ with weight $w(f)$, and bound how many times
  it is charged by other edges $e = (u,v)$. First, consider edges $e$
  on the left of $f$ which charge $f$. If the edge $e$ is located $y$
  hops away from the edge $f$, then this edge $e$ charges $f$ only
  with probability at most $(1-r)^y$. Thus the expected charge from
  all the edges on the left of $f$ is at most
  $$ \sum_{y = 1}^\infty  {(1-r)^y} = \frac{1-r}{r} = O(1)$$
  whenever $p < 1$, since $r = (1-p)^2$. Symmetrically, repeating the
  argument on the right, we conclude that an edge is charged at most
  $O(1)$ times.

  Since we condition on a high probability event $\calF$, the expected
  cost given $\calF$, $\mathbb{E}[\costh(T)\mid\calF]$ remains bounded
  by $O(1) \MST(X)$, as
  $\Pr[\calF] \geq 1 - \nicefrac{1}{n^{10}} \geq 1 - o(1)$.
\end{proof}

\begin{remark}[High-probability guarantees]\label{rem:high-probability}
  In the construction above, a geometric-tail bound and a union bound
  over the edges of $P^*$ show that, with high probability, every repair
  spans $O_p(\log n)$ path edges. Hence each path edge is charged at most
  $O_p(\log n)$ times, giving an uncorrupted spanning tree of cost
  $O_p(\log n)\MST(X)$. In the increase-only model, the observed MST has
  true cost at most its observed cost, which is at most the cost of this
  uncorrupted comparison tree, giving an $O_p(\log n)$-approximation with
  high probability. Using this comparison tree in the decrease-only
  and general algorithms, together with
  \Cref{lemma:dominatinglog,lemma:density,lem:any-almost}, gives an
  $O_p(\log^2 n)$-approximation with high probability in the respective
  parameter ranges.
\end{remark}

\subsection{Finding Spanning Trees in the Increase-Only Model.}
\label{sec:increase-only}

Given the tree-repair lemma, the $O(1)$-approximation algorithm for
the increase-only model is remarkably clean:
\begin{quote}
  \emph{Return a minimum spanning tree \(T\) with respect to the
    corrupted distances \(\widehat d\).}
\end{quote}
\begin{theorem}[Increase-Only Theorem]%
  \label{thm:increase-only-theorem}%
  Let $T$ be the minimum spanning tree with respect to the corrupted
  distances $\dhat$. Then $\EE[\cost(T)] \leq O(1) \cdot \MST(X)$.
\end{theorem}
\begin{proof}
  By~\Cref{lemma:repair}, with probability at least \(1-n^{-10}\),
  there exists a tree \(T'\) consisting only of uncorrupted edges such
  that
  \(\EE[\cost(T')\mid \calF]=\EE[\costh(T')\mid \calF]\le
  O(1)\MST(X)\). Since \(T\) is a minimum spanning tree under the
  corrupted distances \(\widehat d\), we have
  \(\costh(T)\le \costh(T')\). Also, in the increase-only model, the
  observed distances are at least the original distances, so
  \(\cost(T)\le \costh(T)\). Therefore, conditioned on \(\calF\), the
  expected true cost of \(T\) is at most \(O(1)\MST(X)\).

  If event \(\calF\) does not occur, we use the trivial bound
  \(\cost(T)\le (n-1)\Delta\le (n-1)\MST(X)\), where \(\Delta\) is the
  diameter of the metric. Since event \(\calF\) happens with probability at
  least \(1- n^{-10}\), this contributes only \(o(\MST(X))\) to the
  expectation. Hence the expected true cost of the tree returned by
  the algorithm is \(O(1)\MST(X)\).
\end{proof}

\section{Finding Spanning Trees in the Decrease-Only Model.}
\label{sec:decrease-only}

In this section, we present our algorithm for the case where corrupted
distances can only decrease. Since many distances could now have been
set to zero, we cannot do something simplistic like computing an
minimum spanning tree of the corrupted metric. The new idea is to
define the notion of a dominating tree in the corrupted metric, which
we do next.

\subsection{Dominating Trees and the Decrease-Only Algorithm.}
\label{sec:domin-trees}

The crucial definition of the section is that of a dominating
(spanning) tree. Intuitively, this is a spanning tree on the input
point set $X$ which satisfies the triangle inequality \emph{after} the
corruption. Given a tree $T$ and points $u,v \in X$, let
$\langle y_0 = u, y_1, y_2, \ldots, y_{k} = v\rangle$ be the unique
path between $u$ and $v$ in this tree. The (corrupted or observed)
\emph{tree distance} is defined to be $\dhat_T(u,v)$ as
$\sum_{i = 0}^{k-1} \dhat(y_i, y_{i + 1})$.

\begin{definition}[Dominating tree]%
  \label{sec:find-spann-trees}%
  A spanning tree $T$ is \emph{dominating} if
  $\dhat_T(u,v) \geq \dhat(u,v)$ for all $u,v \in X$ with $u \neq v$.
\end{definition}

Note that, for each pair of nodes, the comparison is between the tree
distance between them, and the observed direct distance between them. A
\emph{minimum dominating tree} is a dominating tree $T$ that minimizes
$\costh(T) = \sum_{\{u,v\} \in E(T)} \dhat(u,v)$.  Our algorithm is
very simple once we have the notion of a dominating tree at hand.

\begin{quote}
  \emph{Find any minimum dominating tree $T$, if it exists. If such a tree
  exists, return the path-shortcut $P(T)$. If no dominating
  tree $T$ exists, return any Hamiltonian path $P$ on the point set
  $X$.}
\end{quote}

\subsection{The Analysis.}

\begin{theorem}[Decrease-Only Theorem]%
  \label{thm:decrease-only-theorem}%
  Let $T$ be the minimum-cost dominating tree (or an arbitrary
  Hamiltonian path, if no such tree exists). Then
  $\EE[\,\cost(P(T))\,] \leq O(\log n) \cdot \MST(X)$.
\end{theorem}

We prove \Cref{thm:decrease-only-theorem} using two key lemmas. The
first shows that, conditioned on a high probability event $\calF$
defined below, there exists a dominating tree with small cost with
respect to the corrupted distances. The following lemma shows that for
\emph{the path-shortcut of any dominating tree}, the true cost is at
most $O(\log n)$ times the corrupted cost plus $O(\log n)$ times MST
with high probability.

\begin{lemma}\label{lemma:mindominating}
  There is a high probability event $\calF$ which happens with
  probability at least $ 1- n^{-10}$, so that
  \[ \EE[\,\costh(T) \mid \calF\,] \leq O(1) \MST(X). \]
\end{lemma}

\begin{proof}
  First, using~\Cref{lemma:repair}, we obtain that there is a
  high-probability event $\calF$, and a tree $T'$, such that each
  edge in $T'$ is uncorrupted; moreover,
  $\EE[\,\costh(T') \mid \calF\,] \leq O(1) \MST(X)$.

  Next, observe that since every tree edge of $T'$ is uncorrupted, it
  must be dominating. This follows from the fact that distances only
  decrease after corruption, through the following chain of
  inequalities. Fix a pair $u,v$ of vertices. Then we have
  \[
    \dhat(u,v) \leq d(u,v) \leq d_{T'}(u,v) = \dhat_{T'}(u,v),
  \]
  where the first inequality follows from the decrease-only corruption
  model, the second follows from triangle inequality of the true
  distances, and the last inequality from the fact that the tree $T'$
  does not have any corrupted edges. Finally, $T$ is a minimum
  dominating tree. Thus, by definition,
  $$\EE[\,\costh(T) \mid \calF\,] \leq \EE[\,\costh(T') \mid \calF\,] = O(1) \MST(X)$$
  as desired.
\end{proof}

\begin{lemma}\label{lemma:dominatinglog}
  With probability at least $1 - n^{-10}$ over the
  corruptions, for every dominating tree $T'$, the true cost of $P(T')$
  is at most $O(\log n)\,\costh(T') + O(\log n \cdot \MST(X))$.
\end{lemma}

\begin{proof}
  For brevity, let $P' := P(T')$. For each edge
  $e = \{u,v\} \in E(P')$, let $w(e)$ be the total corrupted weight of
  the tree edges corresponding to the sub-path between $u$ and $v$ on
  the Euler tour used to construct $P'$. Notice that
  $w(P') = \sum_e w(e) \leq 2 \,\costh(T')$ by construction. Also
  observe that for any edge $e = (u,v)$, we have $w(e) \geq \dhat_{T'}(u,v)$, as the length along an Euler tour is at least the shortest path distance in the tree.

  Now we show that the true cost of $P'$ is at most $O(\log n) w(P') + O(\log n) \MST(X)$. To do this, we will upper bound the true distance of each edge, $d(u,v)$, $e = (u,v) \in E(P')$. 

  \textbf{Upper bounding true distance for each edge $e = (u,v)$ of the path $P'$.}
  If $e$ is uncorrupted, observe that $d(u,v) = \dhat(u,v)$. Next, $\dhat(u,v) \leq \dhat_{T'}(u,v) \leq w(e)$ since $T'$ is dominating. 
  Thus, we will use this trivial upper bound $w(e)$. 

  Henceforth, we assume $e$ is corrupted. Assume that $u$ appears before $v$ in the path. Let $\meet(e)$ be the first vertex $x$ to the right of $v$ such that both $u$ and $v$ have an uncorrupted edge incident to $x$. If no such vertex $x$ exists, let $\meet(e) = \perp$.

  Define the upper bound for $d(u,v)$ as $ub(u,v ) = d(u, \meet(e)) + d(v, \meet(e))$ if $\meet(e) \neq \perp$ , and $ub(u,v) = \Delta$ otherwise.\footnote{Recall that $\Delta$ is the maximum (true) distance between two vertices.} %

  First, we show that this is a valid upper bound. If $\meet(e) \neq \perp$, then the upper bound for $d(u,v)$ follows from two facts: first, $(u, \meet(e))$ and $(v, \meet(e))$ are uncorrupted edges; and second, the true metric $d$ satisfies the triangle inequality.

  Otherwise, if $\meet(e) = \perp$, $\Delta$ is simply a valid upper bound for the distance between any pair of vertices. 

  \textbf{Charging upper bounds.} 
  We will now charge the upper bounds from the previous paragraph for each edge of $P'$ to the weights $w(f)$ for some edges $f$ in $P'$. 

  Again, fix the edge $e = (u,v)$. Observe that if for some vertex $x$, the edge $(u,x)$ is an uncorrupted edge, then since $T'$ is dominating, $d(u,x) = \dhat(u,x) \leq \dhat_{T'}(u,x) \leq \sum_{\text{edges $f$ on the $(u,x)$ sub-path in $P'$}} w(f)$. 

  If $e = (u,v)$ is uncorrupted, we simply charge an amount $w(e)$ to the edge $e$. Otherwise, if  $\meet(e)$ exists, by definition, the edge $(u, \meet(e))$ is uncorrupted. We will then charge an amount $2w(f)$ to every edge $f$ on the path between $u$ and $\meet(e)$ in $P'$. Then the total charge will account for $ub(u,v) = d(u,\meet(e)) + d(v,\meet(e))$. Otherwise, simply charge $ub(u,v) = \Delta \leq \MST(X)$ to the $\MST(X)$ cost.

  \textbf{Each edge/MST is not charged too many times.} Fix an edge $f$ on the path, and let us count how many times this edge can be charged. We will do this by bounding the number of edges $e$ that charge $f$. Let us count the number of edges $e$ that lie on the left of $f$ that can charge $f$ (observe that only the edges on the left of $f$ can charge $f$).

  Suppose for some number $t$ that there are at least $2t$ edges of $P'$ to the left of $f$ charging $f$. Let $e_1, e_2 \ldots e_{2t}$ be the edges charging $f$. Of these, pick a set $e_1, e_2 \ldots e_t$ of vertex disjoint edges, and suppose that they appear in this order on $P'$. Note that this is always possible - in a path, among any set of $2t$ edges, there is a matching of size $t$.

  Pick any two of these disjoint edges $e_i = (u_i, v_i)$ and
  $e_j = (u_j, v_j)$ with $i < j$. Since $\meet(e_i)$ was not equal to
  one of the endpoints of $e_j$, in particular, both edges
  $(u_i, v_j)$ and $(v_i, v_j)$ cannot be uncorrupted.

  Consider the event $\calE_{\text{large}}$ which is a superset of this event: there is a set of $2t$ ordered points $y_1, y_2 \ldots y_{2t}$ such that for any $i < j \leq t$, at least one of the edges $(y_{2i - 1}, y_{2j})$ and $(y_{2i}, y_{2j})$ is corrupted. Once we fix the $2t$ points and their ordering, the probability of this is at most $(1 - (1-p)^2)^{t \choose 2}$. There are at most $ {n \choose 2t} (2t!) \leq n^{2t}$ choices for the set of endpoints of these $t$ edges and their ordering. Applying a union bound, the probability that $\calE_{\text{large}}$ holds is at most 

  $$(1 - (1-p)^2)^{t \choose 2} n^{2t} \leq e^{2t \log n - (1 - p)^2{t \choose 2}}.$$

  If $t \geq \frac{30\log n}{(1-p)^2}$, $(1 -p)^2 {t \choose 2} \geq 10t \log n$, and the right hand side is at most $e^{-10\log n} \leq n^{-10}$.

  It follows that with probability at least $1 - n^{-10}$, we must have $t \leq \frac{30 \log n}{(1-p)^2}$. There are at most $2t$ edges that charge $f$, and the charge assigned to $f$ is at most $2w(f)$. Thus the total charge to $f$ is at most $\frac{120\log n}{(1-p)^2}w(f)$.

  Finally, we show that in the case when $\meet(e)$ is not defined, MST is not charged too many times. The proof is exactly identical. Let $e_1, e_2 \ldots e_{2t}$ be the set of edges $e$ for which $\meet(e)$ is not defined. Again, the same argument from the previous paragraph applies for every pair of these edges, showing that $t \leq \frac{30\log n}{(1-p)^2}$ with probability $1 - n^{-10}$.

  Thus the overall cost is at most 
  $$\frac{120\log n}{(1-p)^2} (w(P') + \MST(X))$$
  with high probability, finishing the proof of~\Cref{lemma:dominatinglog}.
\end{proof}

Equipped with these two lemmas, we can now finish the proof of
\Cref{thm:decrease-only-theorem}, that the output tree has true cost
at most $O(\log n) \MST(X)$.

\begin{proof}[Proof of \Cref{thm:decrease-only-theorem}]
  Recall that our algorithm simply outputs a path shortcut of the
  minimum dominating tree if one exists, and otherwise, an arbitrary
  Hamiltonian path. First, we condition on the high probability events
  $\calF$ and $\calE_{\text{large}}$. If either $\calF$ or $\calE_{\text{large}}$
  does not happen, let us trivially bound the true cost of the tree we
  output by $(n - 1)\,\Delta \leq (n-1)\,\MST(X)$. Otherwise, we know
  that there is a minimum dominating tree $T'$ with
  $\EE[\,\costh(T') \mid \calF\,]$ at most $O(1) \MST(X)$.

  Now, let \(X=\costh(T')\). Since \(X\ge 0\),
  \[
    \EE[\,X\mid \calF\cap \calE_{\text{large}}\,] \le \frac{\EE[\,X\mid
      \calF\,]}{\Pr[\,\calE_{\text{large}}\mid \calF\,]}.
  \]
  Moreover, since both \(\calF\) and \(\calE_{\text{large}}\) hold with high
  probability, \(\Pr[\,\calE_{\text{large}}\mid \calF\,]=1-o(1)\). Hence
  \[
    \EE[\,X\mid \calF\cap \calE_{\text{large}}\,] \le O(1)\,\MST(X).
  \]
  Next, by \Cref{lemma:dominatinglog}, the true cost of the path
  shortcut minimum dominating tree $T'$ is at most
  $O(\log n)\, \costh(T') + O(\log n)\,\MST(X)$. Combining and noting
  that $\calE_{\text{large}}$ and $\calF$ together fail with probability at
  most $\nicefrac{1}{n^5}$, we obtain that our output tree (also a path)
  has expected true cost at most $O(\log n)\,\MST(X)$.
\end{proof}

\section{Finding Spanning Trees in the General Setting.}
\label{sec:general-case}

The algorithm for the general case is similar to the decrease-only
case. However, we need to replace the strong requirement of dominating
trees by a weaker requirement, which we call \emph{almost-domination}. The key idea is that while one cannot find a dominating tree where every pair of vertices $u,v$ satisfies the domination property like in the decrease-only case, we can ensure that the domination property holds for most pairs inside \emph{any subset of size at least $C \log n$} for a large enough constant $C$.

\subsection{Dominating Pairs and the Algorithm.}
\label{sec:general-algo}

\begin{definition}[Dominating pair]%
  Given a spanning tree $T$, a pair $(u,v)$ of distinct vertices is a
  dominating pair if $\dhat_T(u,v) \geq \dhat(u,v)$.
\end{definition}

An \emph{almost-dominating tree} is a tree where for any set of
$O(\log n)$ vertices, most pairs inside the set are dominating
pairs. We give the formal definition below.

\begin{definition}[Almost-dominating trees]%
  Suppose $0 \leq p < \frac{1}{4}$ and $\epsilon > 0$ is such that
  $2p + 2\epsilon < \frac{1}{2}$. A spanning tree $T$ is
  almost-dominating if there exists a constant $C_{p,\epsilon}$ so
  that for every set $S$ of vertices with size at least
  $C_{p,\epsilon} \log n$, at least a $1 - p - \epsilon$ fraction of
  the ${S \choose 2}$ pairs are dominating pairs.
\end{definition}

Our algorithm is similar to that for the decrease-only case.

\begin{quote}
  \emph{Compute a minimum almost-dominating tree and return its
    path-shortcut. If no such tree exists, output an arbitrary
    Hamiltonian path.}
\end{quote}

\subsection{The Analysis.}

\begin{theorem}[General Increase-Decrease Theorem]%
  \label{thm:gen-theorem}%
  Let $T$ be the minimum almost dominating tree (or an arbitrary
  Hamiltonian path, if no such tree exists). Then
  $\EE[\,\cost(P(T))\,] \leq O(\log n) \cdot \MST(X)$.
\end{theorem}

The proof first shows that conditional on a high-probability event, there is a tree using only uncorrupted edges that not only has small expected cost as in \Cref{lemma:repair}, but is also almost-dominating. This would then show that the minimum almost-dominating has a small expected corrupted cost. Then, as in the decrease-only case, we prove that this in turn means that the true cost of the path-shortcut of the minimum almost-dominating tree is at most $O(\log n)\MST(X)$ in expectation. 

The following lemma captures a key high-probability event.
\begin{lemma}\label{lemma:density}
  Given $0 \leq p < \frac{1}{4}$ and $\epsilon > 0$, let $\calE$ be the
  event that for every subset $S \subseteq X$ of points of size
  $s \geq C_{p,\epsilon} \log n$, at most a $p + \epsilon$ fraction of
  the ${s \choose 2}$ edges inside $S$ are corrupted. Then there is a
  (large enough) constant $C_{p,\epsilon}$ such that $\calE$ holds
  with probability at least $1 - \frac{1}{n^{10}}$.
\end{lemma}

\begin{proof}
  The case $p=0$ is immediate; assume $p>0$.
  For any set $S$, the expected number of corrupted edges inside $S$
  is at most a $p$ fraction. By a Chernoff bound, the number is within
  $\epsilon p {|S| \choose 2}$ of this value with probability
  $1 - e^{-\epsilon^2 p{|S| \choose 2}/3}$.

  Thus for a sufficiently large value of $C_{p,\epsilon}$, if
  $|S| \geq C_{p,\epsilon} \log n$, the above probability is at least
  $1 - \frac{1}{n^{100|S|}}$. Applying a union bound over all sets
  $S$, we note that there are at most $n^{s}$ subsets of size $s$, and
  there are at most $n$ possible values of $s$. Thus, the event
  $\calE$ occurs with probability at least $1 - \frac{1}{n^{10}}$.
\end{proof}

\begin{lemma}
  \label{lem:cost-min-almost}
  Let $T$ be a minimum almost-dominating tree. Then conditional on a
  high probability event $\calF$,
  $\EE[\,\costh(T) \mid\calF\,] \leq O(1) \MST(X)$.
\end{lemma}

\begin{proof}
  First, we apply~\Cref{lemma:repair} to show that there exists a high
  probability event $\calF_1$, so that if $\calF_1$ occurs, there
  exists a tree $T_0$ with only uncorrupted edges that satisfies
  $\EE[\,\costh(T_0) \mid \calF_1\,] \leq O(1) \MST(X)$.

  Now, let $\calF_2$ be the event that for any set $S$ of size at
  least $C_{p,\epsilon} \log n$ at most a $p + \epsilon$ fraction of
  the internal ${S \choose 2}$ edges are
  corrupted. By~\Cref{lemma:density}, $\calF_2$ holds with probability
  at least $1 - \frac{1}{n^{10}}$.

  This in turn means that
  $\EE[\,\costh(T_0) \mid\calF_1 \cap \calF_2\,] \leq (1 + o(1)) \EE[\,\costh(T_0)
  \mid\calF_1\,] \leq O(1) \MST(X)$. Define $\calF$ to be the
  intersection of events $\calF_1$ and $\calF_2$.
  Since every edge of $T_0$ is uncorrupted, for every uncorrupted pair
  $(u,v)$, we have
  \[
    \dhat_{T_0}(u,v) = d_{T_0}(u,v) \geq d(u,v) = \dhat(u,v).
  \]
  Thus every uncorrupted pair $(u,v)$ is a dominating pair, and since
  in every set of size at least $C_{p,\epsilon} \log n$, at most
  $p + \epsilon$ fraction of the pairs are corrupted, it follows that
  $T_0$ is almost-dominating.

  Finally, since the minimum almost-dominating tree $T$ has
  $\costh(T) \leq \costh(T_0)$ by definition, we have
  \[
    \EE[\,\costh(T) \mid \calF\,] \leq \EE[\,\costh(T_0) \mid \calF\,] \leq O(1) \MST(X).
  \]
\end{proof}

\begin{lemma}
  \label{lem:any-almost}
  There is a high probability event $\calE$ such that, if $\calE$
  happens, then for every almost-dominating tree $T'$, the true cost of
  $P(T')$ is at most $O(\log n)\,\costh(T') + O(\log n \cdot \MST(X))$.
\end{lemma}

\begin{proof}
  Let us write $P'$ as a shorthand for $P(T')$. For each edge $e = \{u,v\} \in E(P')$, let $w(e)$ be the sum of
  the corrupted distances of the tree edges corresponding to the
  sub-path between $u$ and $v$ on the Euler tour used to construct
  $P'$. Notice that $w(P') = \sum_e w(e) \leq 2 \costh(T')$ by
  construction.

  It now suffices to show that the true cost of $P'$ is at most
  $O(\log n) w(P') + O(\log n) \MST(X)$. To do this, we will upper
  bound the true distance of each edge, $d(u,v)$,
  $e = (u,v) \in E(P')$. Call a pair $(x, y)$ good if it is both
  uncorrupted and a dominating pair. Equivalently, a good pair is a
  pair for which
  \[
    d(x,y) = \dhat(x,y) \leq \dhat_{T'}(x,y).
  \]
  Moreover, an edge $e \in P'$ is good if the pair of endpoints is a
  good pair.

  Fix a set $S$ of size $s \geq C_{p,\epsilon} \log n$. Let $\calE$ be
  the event that for every set $S$ of size $s$, at most $p + \epsilon$
  fraction of the ${s \choose 2}$ pairs inside $S$ are
  corrupted. By~\Cref{lemma:density}, $\calE$ holds with probability
  at least $1 - \frac{1}{n^{10}}$. Henceforth assume that $\calE$
  happens. Next, by the definition of domination, at most a
  $p + \epsilon$ fraction of the pairs are not dominating. Thus given
  that $\calE$ happens, at most a $2p + 2\epsilon$ fraction of the
  pairs are bad inside $S$.

  If $e$ is good, a trivial upper bound is $d(u,v) \leq
  w(e)$. Henceforth, we assume $e$ is bad.

  Assume that $u$ appears before $v$ in the path. Let $\meet(e)$ be
  the first vertex $x$ to the right of $v$ such that both pairs
  $(u,x)$ and $(v,x)$ are good. If no such vertex $x$ exists, let
  $\meet(e) = \perp$.

  \textbf{Upper bounding true distance for each edge of the path.}
  Define the upper bound for $d(u,v)$ as
  $ub(u,v ) = d(u, \meet(e)) + d(v, \meet(e))$ if $\meet(e)$ exists,
  and $ub(u,v) = \Delta$, where $\Delta$ is the diameter (maximum true
  distance between two vertices), otherwise.

  First, we show that the inequality $d(u,v) \leq ub(u,v)$ always
  holds. Clearly, if $\meet(e)$ exists, then the upper bound for
  $d(u,v)$ follows from the fact that $(u, \meet(e))$ and
  $(v, \meet(e))$ are good pairs, and the triangle inequality of the
  true metric $d$.

  Second, by definition, $\Delta$ is an upper bound for any
  distance. Hence $d(u,v) \leq \Delta$ holds for any $u,v \in X$. 

  \textbf{Charging upper bounds.} We will now charge the upper bounds
  from the previous paragraph for each edge of $P'$ to the weights
  $w(f)$ for some edges $f$ in $P'$.

  Again, fix the edge $e = (u,v)$. Observe that if for some vertex
  $x$, the edge $(u,x)$ is a good pair,
  $$d(u,x) = \dhat(u,x) \leq \dhat_{T'}(u,x) \leq \sum_{\text{edges $f$
      on the $(u,x)$ sub-path in $P'$}} w(f).$$

  If $e = (u,v)$ is a good pair, we simply charge an amount $w(e)$ to
  the edge $e$. Otherwise, if $\meet(e)$ exists, by definition, the
  pairs $(u, \meet(e))$ and $(v, \meet(e))$ are good. We will then
  charge an amount $2w(f)$ to every edge $f$ on the path between $u$
  and $\meet(e)$ in $P'$. Then the total charge will account for
  $ub(u,v) = d(u,\meet(e)) + d(v,\meet(e))$. Otherwise, simply charge
  $ub(u,v) = \Delta \leq \MST(X)$ to the $\MST(X)$ cost.

  \textbf{Each edge/MST is not charged too many times.} Fix an edge
  $f$, and let us count how many times this edge can be charged. Note
  that $f$ is only charged by edges $e \in E(P')$ to the left of
  $f$. Suppose for some number $t \geq D_{p,\epsilon} \log n$, where
  $D_{p,\epsilon} \geq C_{p,\epsilon}$ is some constant that we will
  set via the analysis, that there are at least $2t$ edges charging
  $f$ viewed in order of the path. Let $e_1, e_2 \ldots e_{2t}$ be the
  edges charging $f$. Of these, suppose without loss of generality by
  reordering that $e_1, e_2 \ldots e_{t}$ is a set of disjoint edges.

  Pick any two of these disjoint edges
  $e_i = (u_i, v_i), e_j = (u_j, v_j), i < j$. Since $\meet(e_i)$ was
  not equal to one of the endpoints of $e_j$, in particular, both
  pairs $(u_i, u_j)$ and $(v_i, u_j)$ cannot be good. Similarly, both
  pairs $(u_i, v_j)$ and $(v_i, v_j)$ cannot be good. Thus for every
  $i < j$ we obtain two bad pairs between the endpoints of the two
  edges $e_i, e_j$. Let $S$ denote the set of the $2t$ endpoints of
  these $t$ edges $e_1, e_2 \ldots e_t$.

  This means that the number of bad pairs inside $S$ is at least
  $2{t \choose 2}$. Normalizing by dividing by ${|S| \choose 2}$, the
  bad pair density is at least
  $$ \frac{2{t \choose 2}}{{2t \choose 2}} = \frac{t-1}{2t-1}.$$

  Recall that $2p + 2\epsilon < \frac{1}{2}$, and observe that the
  above quantity goes to $\frac{1}{2}$ in the limit when $t$ goes to
  infinity. Thus the above quantity is at least $2p + 2\epsilon$ for
  $t \geq D_{p,\epsilon}$, for some constant $D_{p,\epsilon}$
  depending only on $p$ and $\epsilon$. But we already established
  that the high probability event $\calE$, which guarantees that for
  \emph{every} such subset $S$, at most a $2p + 2\epsilon$ fraction of
  the $|S| \choose 2$ pairs inside $S$ are bad. Thus if
  $t \geq C_{p,\epsilon} \log n$, $\calE$ does not happen. Thus,
  $t \leq D_{p,\epsilon} \log n$ with high probability, and given
  this, edge $f$ is charged by at most $O(\log n)$ edges.

  Finally, we show that in the case when $\meet(e)$ is not defined,
  MST is not charged too many times. The proof is exactly
  identical. Suppose there are at least $2t$ such edges. Let
  $e_1, e_2 \ldots e_{2t}$ be the set of edges $e$ for which
  $\meet(e)$ is not defined. Again, extract $t$ disjoint edges in the
  path order. Then the identical argument from the previous paragraph
  shows the same density lower bound of $\frac{t-1}{2t-1}$ for bad
  pairs, which forces $t \leq D_{p,\epsilon} \log n$ with high
  probability. Combining, the overall cost is at most
  $$O(\log n)(w(P') + \MST(X))$$
  finishing the proof of \Cref{lem:any-almost}.
\end{proof}

\begin{proof}[Proof of \Cref{thm:gen-theorem}]
  Let $T$ be the minimum almost-dominating tree computed by the
  algorithm, if one exists. We condition on the high probability event
  $\calF$ from the first lemma and on the high probability event
  $\calE$ from~\Cref{lemma:density}. If either event fails, we use the
  trivial bound $(n-1)\,\Delta\le (n-1)\,\MST(X)$.

  On the good event, the first lemma gives
  \[
    \EE[\,\costh(T)\mid \calF\,]\le O(1)\,\MST(X).
  \]
  Conditioning further on $\calE$ changes this by only a $1+o(1)$
  factor, since $\calE$ holds with high probability. Therefore
  \[
    \EE[\,\costh(T)\mid \calF \cap \calE\,] \le O(1)\,\MST(X).
  \]
  By the second lemma, given $\calF$ and $\calE$ both happen,
  \[
    \cost(P(T))\le O(\log n)\,\costh(T)+O(\log n)\,\MST(X).
  \]
  Taking conditional expectations, we obtain
  \begin{align*}
    \EE[\,\cost(P(T))\mid\calF\cap\calE\,]
      &\le O(\log n)\,\EE[\,\costh(T)\mid\calF\cap\calE\,]
          +O(\log n)\,\MST(X) \\
      &\le O(\log n)\,\MST(X).
  \end{align*}
  Combining these inequalities and using the trivial bound of
  $(n-1)\,\Delta \leq (n-1)\, \MST(X)$ when the good events do not
  occur (which happens with probability at most $\frac{2}{n^{10}}$),
  we obtain
  \[
    \EE[\,\cost(P(T))\,]\le O(\log n)\,\MST(X),
  \]
  as claimed.
\end{proof}

\section{Efficiently Finding Trees in the Decrease-Only  Model.}
\label{sec:efficient-dec}

 In this section, we give an algorithm running in time $n^{O(\log n)}$ that finds a tree $T$ whose expected cost is at most $O(\log^4 n)\MST(X)$ in the decrease-only adversarial model for $p \leq 0.1$. We first assume that the input metric satisfies the property that the distance between any two points is at least $1$ and at most $\Delta$. We also assume that we already know some number $\Delta \leq \Delta' \leq 2\Delta$. In this setting, we give an $O(\log^2 n \log^2 \Delta')$ approximation. Later, at the end of this section we show how to reduce to the case where $\Delta = n^2$, and also remove the two assumptions - the minimum distance being at least $1$, and the algorithm knowing $\Delta'$.

 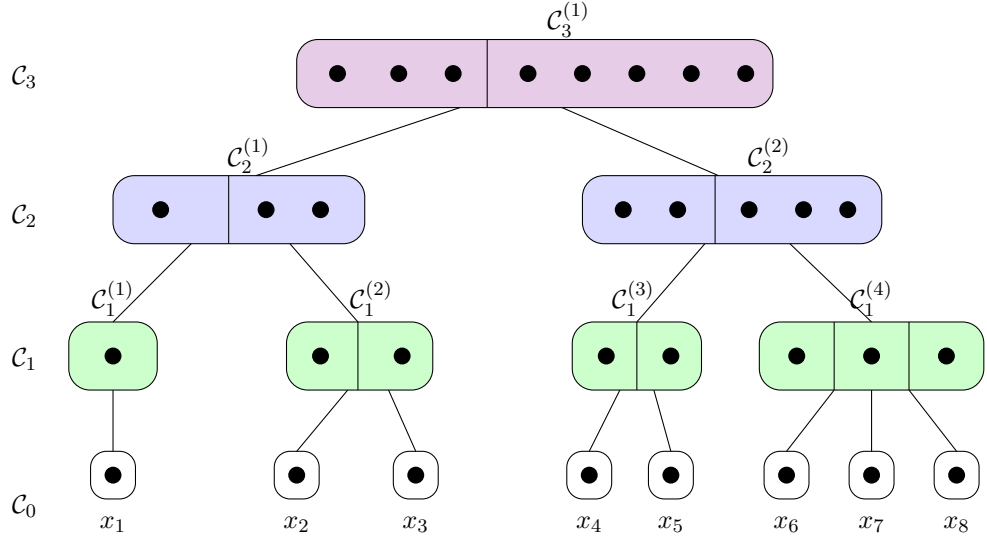
\begin{figure}[!htbp]
\centering
\begin{tikzpicture}[scale=0.9, every node/.style={transform shape}]

\tikzset{
  c3/.style={draw, rounded corners=8pt, fill=violet!20, minimum height=1cm},
  c2/.style={draw, rounded corners=8pt, fill=blue!15, minimum height=1cm},
  c1/.style={draw, rounded corners=8pt, fill=green!20, minimum height=1cm},
  leaf/.style={draw, rounded corners=6pt, minimum width=.7cm, minimum height=.7cm},
  bdot/.style={circle, fill=black, inner sep=2.5pt},
  rdot/.style={circle, fill=black, inner sep=2.5pt}
}

\path[use as bounding box] (-4.7,-1.8) rectangle (10.5,6.25);

\node at (-4,5.15) {$\mathcal{C}_3$};
\node at (-4,3.10) {$\mathcal{C}_2$};
\node at (-4,1.00) {$\mathcal{C}_1$};
\node at (-4,-1.15) {$\mathcal{C}_0$};

\draw[c3] (0,4.7) rectangle (7,5.7);
\draw (2.8,4.7) -- (2.8,5.7);
\node at (4.0,6.0) {$\mathcal{C}_3^{(1)}$};

\node[bdot] at (0.6,5.2) {};
\node[bdot] at (1.5,5.2) {};
\node[bdot] at (2.3,5.2) {};

\coordinate (rtop) at (3.4,5.2);
\node[rdot] at (rtop) {};
\node[bdot] at (4.2,5.2) {};
\node[bdot] at (5.0,5.2) {};
\node[bdot] at (5.8,5.2) {};
\node[bdot] at (6.6,5.2) {};

\draw[c2] (-2.7,2.7) rectangle (1.0,3.7);
\draw (-1.0,2.7) -- (-1.0,3.7);
\node at (-0.7,4.0) {$\mathcal{C}_2^{(1)}$};
\node[bdot] at (-2.0,3.2) {};
\coordinate (r21) at (-0.45,3.2);
\node[rdot] at (r21) {};
\node[bdot] at (0.35,3.2) {};

\draw[c2] (4.2,2.7) rectangle (8.6,3.7);
\draw (6.15,2.7) -- (6.15,3.7);
\node at (6.95,4.0) {$\mathcal{C}_2^{(2)}$};
\node[bdot] at (4.8,3.2) {};
\node[bdot] at (5.6,3.2) {};
\coordinate (r22) at (6.65,3.2);
\node[rdot] at (r22) {};
\node[bdot] at (7.45,3.2) {};
\node[bdot] at (8.1,3.2) {};

\draw[c1] (-3.35,0.55) rectangle (-2.05,1.55);
\node at (-2.70,1.9) {$\mathcal{C}_1^{(1)}$};
\coordinate (r11) at (-2.70,1.05);
\node[rdot] at (r11) {};

\draw[c1] (-0.15,0.55) rectangle (2.0,1.55);
\draw (0.9,0.55) -- (0.9,1.55);
\node at (1.1,1.9) {$\mathcal{C}_1^{(2)}$};
\node[bdot] at (0.35,1.05) {};
\coordinate (r12) at (1.55,1.05);
\node[rdot] at (r12) {};

\draw[c1] (4.05,0.55) rectangle (5.95,1.55);
\draw (5.0,0.55) -- (5.0,1.55);
\node at (4.95,1.9) {$\mathcal{C}_1^{(3)}$};
\coordinate (r13) at (4.55,1.05);
\node[rdot] at (r13) {};
\node[bdot] at (5.5,1.05) {};

\draw[c1] (6.8,0.55) rectangle (10.1,1.55);
\draw (7.9,0.55) -- (7.9,1.55);
\draw (9.0,0.55) -- (9.0,1.55);
\node at (8.45,1.9) {$\mathcal{C}_1^{(4)}$};
\node[bdot] at (7.35,1.05) {};
\coordinate (r14) at (8.45,1.05);
\node[rdot] at (r14) {};
\node[bdot] at (9.55,1.05) {};

\foreach \x/\lab/\name in {
  -2.70/1/b1,
   0.00/2/b2,
   1.75/3/b3,
   4.30/4/b4,
   5.50/5/b5,
   7.20/6/b6,
   8.45/7/b7,
   9.70/8/b8}
{
  \draw[leaf] (\x-0.33,-1.05) rectangle (\x+0.33,-0.35);
  \coordinate (\name) at (\x,-0.70);
  \node[rdot] at (\name) {};
  \node at (\x,-1.42) {$x_{\lab}$};
}

\draw (2.40,4.70) -- (-0.60,3.70);
\draw (3.90,4.70) -- (6.20,3.70);

\draw (-1.55,2.70) -- (-2.70,1.55);
\draw (-0.10,2.70) -- (0.90,1.55);

\draw (6.00,2.70) -- (5.00,1.55);
\draw (7.25,2.70) -- (8.45,1.55);

\draw (-2.70,0.55) -- (-2.70,-0.35);

\draw (0.75,0.55) -- (0.00,-0.35);
\draw (1.35,0.55) -- (1.75,-0.35);

\draw (4.75,0.55) -- (4.30,-0.35);
\draw (5.25,0.55) -- (5.50,-0.35);

\draw (7.90,0.55) -- (7.20,-0.35);
\draw (8.45,0.55) -- (8.45,-0.35);
\draw (9.00,0.55) -- (9.70,-0.35);

\end{tikzpicture}
\caption{The structure of a good hierarchical partition on 8 points. }
\label{fig:cluster-tree}
\end{figure}

 \paragraph{Roadmap.} This section is divided into three subsections, defined below. The main object of this section is a \emph{good hierarchical partition}. In the first subsection, we show that once we compute a good hierarchical partition, we can obtain our efficient approximation algorithm. In the second subsection, we then show how to compute a good hierarchical partition efficiently. In the final subsection, we remove the assumptions stated above. 

\begin{definition}[Good Hierarchical Partition]%
A sequence of partitions $\mathcal{H} = (\mathcal{C}_0, \mathcal{C}_1, \mathcal{C}_2 \ldots \mathcal{C}_{\ell})$, $\ell := \lceil \log_2 \Delta' \rceil$ is a good hierarchical partition if it satisfies the following properties. 

\begin{itemize}

\item Each $\mathcal{C}_{i-1}$ is a refinement of $\mathcal{C}_i$, $\mathcal{C}_{0}$ consists only of singletons, and $\mathcal{C}_{\ell} = \{X\}$, the set of all points.

\item There exist absolute constants $c_{out}$ and $\beta$, such that for each \(C\in\mathcal C_i\), called a cluster, either cluster \(C\) is small, meaning
\[
|C|\le 3c_{\mathrm{out}}\log n,
\]
or there is a set \(\mathrm{Out}(C)\subseteq C\) of size at most
\(c_{\mathrm{out}}\log n\) such that
\[
\mathrm{diam}(C\setminus \mathrm{Out}(C))\le \beta 2^i.
\]
We refer to \(C\setminus \mathrm{Out}(C)\) as \(\mathrm{core}(C)\).  For a
small non-root cluster, we set \(\mathrm{core}(C)=\emptyset\) and
\(\mathrm{Out}(C)=C\).  For the root cluster \(X\in\mathcal C_\ell\), we set
\(\mathrm{core}(X)=X\).

\item The cost of the hierarchical partition $\cost(\mathcal{H})$ defined as $\sum_{i = 0}^\ell 2^i n_i$, where $n_i = |\mathcal{C}_i|$, is at most $O(MST \log n \log \Delta')$.  

\end{itemize}

\end{definition}

 \subsection{Approximation algorithm given a good hierarchical partition.}
\begin{theorem}\label{thm:rounding}
Given a good hierarchical partition $\mathcal{H}$, there is a polynomial time randomized algorithm that, with high probability over the corruptions, computes a spanning tree $T'$, whose expected true cost is at most $O(\log n \log \Delta' \cdot \cost(\mathcal{H}))$.\footnote{Note that the expectation is over the randomness of the algorithm, while high probability is over the corruptions.} 
\end{theorem}

\begin{proof}

The following algorithm shows how to compute the tree $T'$.
\begin{algorithm}
\caption{Computing a tree from a good hierarchical partition.}
\label{alg:rounding}
\begin{algorithmic}
\State Given the hierarchical partition $\mathcal{C}_0,\mathcal{C}_1, \ldots \mathcal{C}_{\ell}$, for each $\mathcal{C}_i$ and each part $C \in \mathcal{C}_i$, choose a leader $p_C$ uniformly at random from $C$.
\State Connect them naturally. There is an edge between $p_C$ and $p_{C'
}$ if $C' \subseteq C$, and $C' \in \mathcal{C}_{i}$ and $C \in \mathcal{C}_{i+1}$. Let the resulting tree be $T$.

\State Let $T' = P(T)$ be the path-shortcut of $T$, obtained by listing the vertices in the order in which they first appear in an Euler tour.

\State Output $T'$.

\end{algorithmic}

\end{algorithm}

\begin{figure}[t]
\centering
\begin{tikzpicture}[scale=0.9, every node/.style={transform shape}]

\tikzset{
  c3/.style={draw, rounded corners=8pt, fill=violet!20, minimum height=1cm},
  c2/.style={draw, rounded corners=8pt, fill=blue!15, minimum height=1cm},
  c1/.style={draw, rounded corners=8pt, fill=green!20, minimum height=1cm},
  leaf/.style={draw, rounded corners=6pt, minimum width=.7cm, minimum height=.7cm},
  bdot/.style={circle, fill=black, inner sep=2.5pt},
  rdot/.style={circle, fill=red, draw=white, inner sep=3pt}
}

\path[use as bounding box] (-4.7,-1.8) rectangle (10.5,6.25);

\node at (-4,5.15) {$\mathcal{C}_3$};
\node at (-4,3.10) {$\mathcal{C}_2$};
\node at (-4,1.00) {$\mathcal{C}_1$};
\node at (-4,-1.15) {$\mathcal{C}_0$};

\draw[c3] (0,4.7) rectangle (7,5.7);
\draw (2.8,4.7) -- (2.8,5.7);
\node at (4.0,6.0) {$\mathcal{C}_3^{(1)}$};

\node[bdot] at (0.6,5.2) {};
\node[bdot] at (1.5,5.2) {};
\node[bdot] at (2.3,5.2) {};

\coordinate (rtop) at (3.4,5.2);
\node[rdot] at (rtop) {};
\node[bdot] at (4.2,5.2) {};
\node[bdot] at (5.0,5.2) {};
\node[bdot] at (5.8,5.2) {};
\node[bdot] at (6.6,5.2) {};

\draw[c2] (-2.7,2.7) rectangle (1.0,3.7);
\draw (-1.0,2.7) -- (-1.0,3.7);
\node at (-0.7,4.0) {$\mathcal{C}_2^{(1)}$};
\node[bdot] at (-2.0,3.2) {};
\coordinate (r21) at (-0.45,3.2);
\node[rdot] at (r21) {};
\node[bdot] at (0.35,3.2) {};

\draw[c2] (4.2,2.7) rectangle (8.6,3.7);
\draw (6.15,2.7) -- (6.15,3.7);
\node at (6.95,4.0) {$\mathcal{C}_2^{(2)}$};
\node[bdot] at (4.8,3.2) {};
\node[bdot] at (5.6,3.2) {};
\coordinate (r22) at (6.65,3.2);
\node[rdot] at (r22) {};
\node[bdot] at (7.45,3.2) {};
\node[bdot] at (8.1,3.2) {};

\draw[c1] (-3.35,0.55) rectangle (-2.05,1.55);
\node at (-2.70,1.9) {$\mathcal{C}_1^{(1)}$};
\coordinate (r11) at (-2.70,1.05);
\node[rdot] at (r11) {};

\draw[c1] (-0.15,0.55) rectangle (2.0,1.55);
\draw (0.9,0.55) -- (0.9,1.55);
\node at (1.1,1.9) {$\mathcal{C}_1^{(2)}$};
\node[bdot] at (0.35,1.05) {};
\coordinate (r12) at (1.55,1.05);
\node[rdot] at (r12) {};

\draw[c1] (4.05,0.55) rectangle (5.95,1.55);
\draw (5.0,0.55) -- (5.0,1.55);
\node at (4.95,1.9) {$\mathcal{C}_1^{(3)}$};
\coordinate (r13) at (4.55,1.05);
\node[rdot] at (r13) {};
\node[bdot] at (5.5,1.05) {};

\draw[c1] (6.8,0.55) rectangle (10.1,1.55);
\draw (7.9,0.55) -- (7.9,1.55);
\draw (9.0,0.55) -- (9.0,1.55);
\node at (8.45,1.9) {$\mathcal{C}_1^{(4)}$};
\node[bdot] at (7.35,1.05) {};
\coordinate (r14) at (8.45,1.05);
\node[rdot] at (r14) {};
\node[bdot] at (9.55,1.05) {};

\foreach \x/\lab/\name in {
  -2.70/1/b1,
   0.00/2/b2,
   1.75/3/b3,
   4.30/4/b4,
   5.50/5/b5,
   7.20/6/b6,
   8.45/7/b7,
   9.70/8/b8}
{
  \draw[leaf] (\x-0.33,-1.05) rectangle (\x+0.33,-0.35);
  \coordinate (\name) at (\x,-0.70);
  \node[rdot] at (\name) {};
  \node at (\x,-1.42) {$x_{\lab}$};
}

\draw (rtop) -- (r21);
\draw (rtop) -- (r22);

\draw (r21) -- (r11);
\draw (r21) -- (r12);

\draw (r22) -- (r13);
\draw (r22) -- (r14);

\draw (r11) -- (b1);

\draw (r12) -- (b2);
\draw (r12) -- (b3);

\draw (r13) -- (b4);
\draw (r13) -- (b5);

\draw (r14) -- (b6);
\draw (r14) -- (b7);
\draw (r14) -- (b8);

\end{tikzpicture}
\caption{The hierarchical leader tree $T$ constructed by \Cref{alg:rounding}, obtained by sampling a point from each cluster and connecting naturally along the hierarchy.}
\label{fig:cluster-tree-leader}
\end{figure}
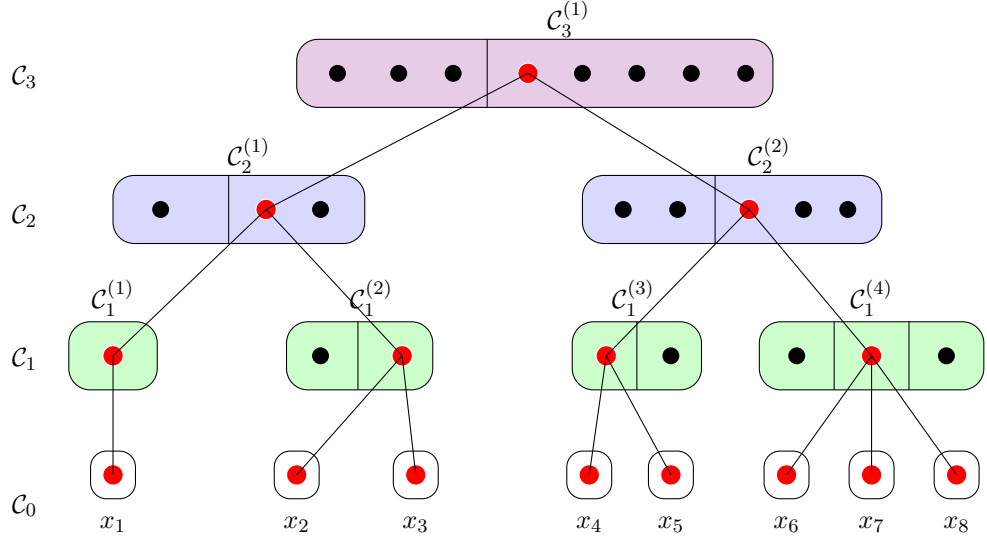

Note that $T'$ is spanning as the set of leaves of $T$ is exactly the entire point set $X$. Now let us prove that the tree $T'$ has small cost in expectation. It is sufficient to bound the cost of $T$, as the last step loses at most a factor of $2$.

\noindent\textbf{Charging scheme.} Recall that the cost of scale $t = 2^i$ appearing in $cost(\mathcal{H})$ is $2^i n_i$, where $n_i$ is the number of clusters at scale $2^i$. Let us view this by associating with each cluster at this scale a cost of $t$. The total cost of this level is the sum of these associated costs. 
We will upper bound the total cost of $T$ by charging each edge of the hierarchical leader-tree
to a \emph{single} cluster in a partition on a suitable level, and then showing that each cluster is charged only
$O(\log n \log \Delta)$ times its own cost in expectation.

Consider an edge $(u,v)\in T$ where $u$ is the leader of a parent cluster at level $i+1$ and
$v$ is the leader of a child cluster at level $i$ (so $C_v\subseteq C_u$ where
$C_u\in\mathcal C_{i+1}$ and $C_v\in\mathcal C_i$).

For a point $x\in X$ and an index $k$, let $C_k(x)\in\mathcal C_k$ denote the unique cluster
containing $x$ at level $k$.

For \(x\in\{u,v\}\), define the promotion level of \(x\) relative to the edge $e = (u,v)$ as follows.
\[
\tau_e(x):=\min\{k\in\{i+1,i+2,\ldots,\ell\}:x\in \mathrm{core}(C_k(x))\}.
\]

This is well-defined because \(C_\ell(x)=X\) and \(\mathrm{core}(X)=X\).  Let
\[
k^\star:=\max\{\tau_e(u),\tau_e(v)\}.
\]

\noindent\textbf{Distance bound for an edge.} Let \(D:=C_{k^\star}(u)=C_{k^\star}(v)\). Note that $u$ and $v$ must be contained in the same cluster at level $k^\star$, since they are contained in the same cluster at level $i + 1$, and $k^\star \geq i + 1$.

Let
\[
A_u:=C_{\tau_e(u)}(u),\qquad A_v:=C_{\tau_e(v)}(v).
\]
By definition, \(u\in \mathrm{core}(A_u)\) and
\(v\in \mathrm{core}(A_v)\), with \(A_u,A_v\subseteq D\).
We claim that $d(u,v)\le O(2^{k^\star})$. Consider \(x\in\{u,v\}\), and
let \(A_x=C_{\tau_e(x)}(x)\).  Since \(x\in \mathrm{core}(A_x)\), and
\(A_x\subseteq D\), it is enough to find a point in
\(\mathrm{core}(A_x)\cap \mathrm{core}(D)\). We can then use the fact that cores have small diameter. First observe that neither $A_x$ nor $D$ is a small cluster, since by definition, small clusters have empty cores. Now \(D\) has at most \(c_{\mathrm{out}}\log n\) outliers, while
\(|\mathrm{core}(A_x)|>c_{\mathrm{out}}\log n\), and $A_x \subseteq D$.  Hence
\[
\mathrm{core}(A_x)\cap \mathrm{core}(D)\neq \emptyset.
\]
Pick \(w_x\in \mathrm{core}(A_x)\cap \mathrm{core}(D)\).
Applying this to \(u\) and \(v\), we get points \(w_u,w_v\in \mathrm{core}(D)\)
with \(w_u\in\mathrm{core}(A_u)\) and \(w_v\in\mathrm{core}(A_v)\). Therefore
\[
d(u,v)
\le d(u,w_u)+d(w_u,w_v)+d(w_v,v)
\le \beta 2^{\tau_e(u)}+\beta 2^{k^\star}+\beta 2^{\tau_e(v)}
=O(\beta 2^{k^\star}).
\]

\medskip
\noindent\textbf{Where we charge.} 
If \(k^\star=i + 1\), then both endpoints lie in the core of the parent cluster,
so
\[
d(u,v)\le \beta 2^{i+1}.
\]
In this case, charge the edge \((u,v)\) to the child cluster \(C_v\in
\mathcal C_i\).

Otherwise \(k^\star>i + 1\).  Choose an endpoint \(x\in\{u,v\}\) with
\(\tau_e(x)=k^\star\), breaking ties arbitrarily, and charge the edge to
$C_{k^\star-1}(x)$. Since \(x\notin \mathrm{core}(C_{k^\star-1}(x))\) but
\(x\in \mathrm{core}(C_{k^\star}(x))\), the endpoint \(x\) is an outlier of
the charged cluster.  Moreover the charge is \(O(2^{k^\star})\), which is only
a constant factor larger than the scale \(2^{k^\star-1}\) of the charged
cluster.

\medskip
\noindent\textbf{Bounding expected charge to one cluster.}
Fix any cluster $C\in\mathcal C_j$.
Under the above rule, every edge charged to $C$ has charge at most $O(2^{j+1})=O(2^j)$.

Moreover, $C$ can receive charge in only two ways:
(i) a single edge  corresponding to the unique parent--child link in the cluster-tree, where one endpoint is in $C$ and the other endpoint is in its parent cluster. There is at most one such edge per cluster.
(ii) edges where one endpoint is an \emph{outlier vertex of $C$}.
Thus,
\[
\text{(total charge to $C$)}
\ \le\ O(2^j)\;+\;\sum_{x\in \mathrm{Out}(C)} O(2^j)\cdot N_x(C),
\]
where $N_x(C)$ is the number of hierarchical-tree edges incident to $x$ that end up being
charged to $C$ (note each such edge is incident to $x$ by construction).

We now bound $\EE[N_x(C)]$ by a crude but sufficient upper bound: $N_x(C)\ \le\ \deg_{\mathrm{hier}}(x)$,
where $\deg_{\mathrm{hier}}(x)$ is the total number of edges in the hierarchical leader-tree
incident to $x$ (across all levels).
For each level $r$, $x$ is chosen as the leader of its unique cluster $C_r(x)$ with probability
$1/|C_r(x)|$; conditioned on this event, the number of incident edges contributed at that level
is at most $1 + \#\{\text{children of }C_r(x)\}\le 1+|C_r(x)|$.
Hence the expected number of incident edges contributed by level $r$ is $O(1)$, and summing over
$r=0,1,\dots,\ell$ gives
$\EE[\deg_{\mathrm{hier}}(x)]\ \le\ O(\ell)\ =\ O(\log \Delta)$. Since $|\mathrm{Out}(C)|\le 3c_{out}\log n$, we conclude
\[
\EE[\text{charge}(C)]
\ \le\ O(2^j)\;+\;|\mathrm{Out}(C)|\cdot O(2^j)\cdot O(\log\Delta)
\ =\ O(2^j\log n\log\Delta).
\]
Since the cost we assigned to $C$ is $2^j$, this is $O(\log n\log\Delta')$ times the cluster's cost.
Summing over all clusters and all levels yields
$\EE[\mathrm{cost}(T)] \le O(\log n\log\Delta')\cdot \mathrm{cost}(\mathcal H)$.
\end{proof}

\subsection{Computing a good hierarchical partition.}
Next, we describe the algorithm to compute such a good hierarchical partition with high probability. We start with a key lemma and then give some important definitions.

\begin{lemma}[Core-Outlier lemma]\label{lemma:coreoutliereff}
With high probability over the corruption, for any set $S$ of $100\log n$ points that forms a clique in $\widehat{G}_t$ for some scale $t$, there exists a core $Y$ of at least $80\log n$ vertices, with $Y \subseteq S$, such that $Y$ has diameter at most $\alpha t$ where $\alpha = 2$.
\end{lemma}

\begin{proof}
Let \(L=\log n\), and write \(m=100L\).  We first prove the following event:
for every set \(S\subseteq X\) of size \(m\), there is a vertex \(z\in S\)
incident to at most \(20L\) corrupted pairs inside \(S\).

Fix a set \(S\) of size \(m\), and let \(Z_S\) be the number of corrupted pairs
inside \(S\).  Since the corruptions are independent and each pair is corrupted
with probability at most \(0.1\), \(Z_S\) is stochastically dominated by
\[
B\sim \mathrm{Bin}\left(\binom{m}{2},0.1\right).
\]
Let
\[
\mu:=\EE[B]
=0.1\binom{100L}{2}
=0.1\cdot \frac{100L(100L-1)}{2}
=500L^2-5L.
\]
For \(L\ge 1\), we have $495L^2\le \mu\le 500L^2$.
If every vertex of \(S\) is incident to more than \(20L\) corrupted pairs inside
\(S\), then summing corrupted degrees over all vertices of \(S\) gives
$2Z_S>m\cdot 20L=100L\cdot 20L=2000L^2$. Hence $Z_S>1000L^2$.

Since \(Z_S\) is stochastically dominated by \(B\), and since
\(1000L^2\ge 2\mu\), we have
\[
\Pr[Z_S>1000L^2]
\le \Pr[B>1000L^2]
\le \Pr[B\ge 2\mu].
\]
By the standard Chernoff bound, $\Pr[B\ge 2\mu]\le \exp(-\mu/3)$. 
Using \(\mu\ge 495L^2\), we get
\[
\Pr[Z_S>1000L^2]\le \exp(-165L^2).
\]

Now union bound over all choices of \(S\).  Since \(L=\log_2 n\), the number
of \(m\)-subsets of \(X\) is at most
\[
\binom{n}{m}\le n^m=n^{100L}=2^{100L^2}
=\exp(100\ln(2)L^2)\le \exp(70L^2).
\]
Thus the probability that there exists a set \(S\) of size \(100L\) for which
every vertex is incident to more than \(20L\) corrupted pairs is at most $\exp(70L^2)\cdot \exp(-165L^2)
=\exp(-95L^2)$.
Hence, with probability at least \(1-\exp(-95L^2)\), every set \(S\) of size
\(100L\) contains a vertex \(z\in S\) incident to at most \(20L\) corrupted
pairs inside \(S\).

Condition on this event.  Let \(S\) be any set of \(100L\) points that forms a
clique in \(\widehat G_t\).  Choose \(z\in S\) incident to at most \(20L\) corrupted
pairs inside \(S\), and define
\[
Y:=\{z\}\cup\{y\in S\setminus\{z\}:(z,y)\text{ is uncorrupted}\}.
\]
Since \(z\) has at most \(20L\) corrupted incident pairs inside \(S\), $|Y|\ge 1+(100L-1-20L)=80L$.
For every \(y\in Y\), the pair \((z,y)\) is uncorrupted.  Since \(S\) is a
clique in \(\widehat G_t\), the observed distance between \(z\) and \(y\) is at
most \(t\).  Because \((z,y)\) is uncorrupted, its true distance is also at
most \(t\): $d(z,y)\le t$.
Therefore, for any \(y,y'\in Y\), $d(y,y')\le d(y,z)+d(z,y')\le 2t$.
Thus \(Y\) has true diameter at most \(2t\), and \(|Y|\ge 80\log_2 n\).  This
proves the lemma with \(\alpha=2\).
\end{proof}

\begin{definition}[Distance from a point to a set]%
Let $S$ be a set with $|S| =  100\log n$, and let $p$ be a point. We say that the distance from $p$ to $S$ is $r$ if $r$ is the smallest value for which there are at least $80\log n$ edges in $\widehat{G}_r$ (the corrupted graph with edge lengths at most $r$) between $p$ and $S$. We denote this distance by $\widehat{d}(p,S)$.
\end{definition}

\begin{definition}[Distance between two sets]%
Let $S,S'$ be disjoint sets with $|S|,|S'| = 100\log n$. We say that the distance between $S$ and $S'$ is $r$ if $r$ is the smallest value for which there are at least $5000\log^2 n$ edges in $\widehat{G}_r$ between $S$ and $S'$. We denote this quantity by $\widehat{d}(S,S')$.\footnote{For convenience, we also define $\widehat{d}(S,S)=0$.}
\end{definition}
\begin{definition}[Admissible set]%
Fix some scale $t = 2^i$. A set of \(100\log n\) points \(S\) is admissible in \(\widehat G_t\) if it forms a
clique in \(\widehat G_t\), and furthermore, for all sets \(S'\) disjoint from
\(S\) that form a clique in \(\widehat G_{t'}\) for \(t' =  2^j\) for some $\ell \geq j \geq i$, and every
\(p\in S\),
\[
\widehat d(p,S')\le 10\alpha t+\widehat d(S,S')+\alpha t'.
\]
\end{definition}
Intuitively, admissible sets $S$ are points that are close to each other after corruption, but each point must also be close to other sets that are close to $S$. As we will see later, it is important that the definition of admissible sets performs this ``triangle-inequality test'' on the distances.

Algorithm~\ref{alg:hierarchical-partition} gives the algorithm to construct a good hierarchical partition. First, we explain what the algorithm does on a high level. In this overview we omit some details for ease of understanding.

\paragraph{Overview of the algorithm.} The algorithm works bottom up, across scales $t = 2^i$, starting from $i = 1$ till $i = \ell-1$. At level $i$, we will compute the partition $\mathcal{C}_i$ of $X$. In each level, we will create clusters $C \in \mathcal{C}_i$. Each cluster $C \in \mathcal{C}_i$ is created by first obtaining a seed set $S$ of size $100 \log n$, which forms a clique at scale $2^i$. Then, more points may be added to this cluster via the \textsc{Carve} operation. Intuitively, the \textsc{Carve} operation finds points that have many edges at scale $t$ to the set $S$: so that we expect most of them to be close to $S$. Once the \textsc{Carve} operation is completed, we obtain a cluster $C$ generated by the seed set $S$.

To make sure that the clustering is hierarchical, we operate as follows. Suppose we have already computed $\mathcal{C}_{i-1}$. Then to compute $\mathcal{C}_i$, we do not work with the entire point set, but only work with \emph{representatives} from $\mathcal{C}_{i-1}$. This set of representatives is captured by the set $\mathcal{L}$. It consists of two types of objects. First, it consists of seed sets of size $100\log n$ that generated the clusters of $\mathcal{C}_{i-1}$. Second, it consists of singleton points which form singleton clusters at level $i - 1$ (because they were far apart from other points at scale $2^{i-1}$). Then the candidates for seed sets $S$ at level $i$ are either the seed sets from level $i  - 1$, or a set of $100\log n$ fresh singletons. We pick such a candidate seed set $S$, carve out objects from $\mathcal{L}$, and then do a \emph{natural hierarchical extension}. This simply means that if we carved out a seed from level $i - 1$, we add the whole cluster of $\mathcal{C}_{i-1}$ that it represented to the current cluster. This ensures that $\mathcal{C}_i$ remains a partition of the whole point set, and $\mathcal{C}_{i-1}$ is a refinement of $\mathcal{C}_i$. 

Next, we describe the intuition behind the $\textsc{Carve}{}$ operation and the requirement that seed sets be admissible. Since $S$ is a clique in the corrupted graph at the appropriate scale $t$, by the core outlier lemma (\Cref{lemma:coreoutliereff}), $S$ contains many vertices which are truly close at scale $t$. The idea is, given a seed set $S$, we want to carve all the points $p$ which are truly close to $S$. Thus, if $p$ is close to $S$, we will see many edges from $p$ to $S$ at scale $t$. However, we do not want to select too many outliers. While via a simple probability analysis, it is clear that there will be at most $O(\log n)$ outliers which are far from $S$ but have many edges to $S$ after corruption, the danger is \emph{outlier accumulation}. When two clusters at a level merge, their outliers may merge, leading to an explosion of outliers as we go up the hierarchy. To avoid this, while carving with $S$, we enumerate all possible future seed candidates $S'$ at a scale $t' \geq t$, and demand that if $p$ is observed to be close to $S$, and $S$ is observed to be close to $S'$, then $p$ has many observed edges to $S'$ as well. This ensures that even at a higher level of the hierarchy, for any cluster created by a seed set $S'$, every point $p$, even coming from a lower level cluster, has many edges to $S'$. This will force $O(\log n)$ outliers for any cluster, as desired (see \Cref{lemma:diameterhier}.)

\begin{algorithm}[!htbp]
\caption{Building Hierarchical Partition}
\label{alg:hierarchical-partition}
\begin{algorithmic}
\State \textbf{Input:} The point set $X$.
\State \textbf{Output:} A good hierarchical partition.
\State Set \(\mathcal C_0\) to be the partition of \(X\) with each point as a singleton.
\State Initialize \(\mathcal L\) to be the set of all singleton points.
\For{scale $t = 2^i$, $i = 1,2,\ldots,\ell - 1$}
\State  Initialize \(\mathcal S_i\gets \emptyset\),  $\mathcal{C}_i \gets \emptyset$.
\While{there exists a seed in \(\mathcal L\), or an
admissible set of \(100\log n\) singleton points in \(\mathcal L\)}
    \State If there exists a seed \(S\in\mathcal L\),
    choose one arbitrarily. Otherwise choose an admissible set \(S\) of
    \(100\log n\) singleton points from \(\mathcal L\).
    \State Add \(S\) to \(\mathcal S_i\).
    \State $C \gets \textsc{Carve}(S,t,\mathcal{L})$
    \State  $C$ is now formed as a cluster with seed set $S$. 
    \State Mark all points in $C$ as carved.
    \State Remove all sets containing points of $C$ from $\mathcal{L}$.

    \State \textbf{(Natural hierarchical extension)} For every seed $S'$ that is in $C$, add the cluster of $S'$ from the previous level to $C$.
    \State Let $C$ be the final resulting cluster. We say that $C$ is the cluster of $S$. Add $C$ to $\mathcal{C}_i$.
    
    \EndWhile
    \State Add every remaining singleton in \(\mathcal L\) as a singleton cluster to
\(\mathcal C_i\).
    \State Update $\mathcal{L}$ to be the union of $\mathcal{S}_i$ (the set of seeds) and the unclustered singleton points.

\EndFor
\State Set \(\mathcal C_{\ell} \gets \{X\}\).
\end{algorithmic}

\end{algorithm}
\begin{algorithm}[!htbp]
\caption{\textsc{Carve}($S,t = 2^i,\mathcal{L}$)}
\begin{algorithmic}
\State \textbf{Input:} $S,t = 2^i,\mathcal{L}$
\State \textbf{Output:} A set $Y$ of points that represents the carved cluster.
\State Initialize \(Y\gets \emptyset\).
\State If \(S\in\mathcal L\) is an old representative seed, let $Y \gets Y \cup S$.
\State For all other seeds \(S'\in\mathcal L\), \(S'\neq S\), that are at
distance at most \(10\alpha t\) from \(S\), do $Y \gets Y \cup S'$.

\State Enumerate all pairs \((S',t')\), where $t' = 2^j$, and $\ell \geq j \geq i$ and \(S'\) is a set of size \(100\log n\), disjoint from \(S\), that
forms a clique in \(\widehat G_{t'}\).

\State For each singleton point $p \in \mathcal{L}$, set $Y \gets Y \cup \{p\}$ iff $\widehat{d}(p,S) \leq 10\alpha t$, and for every $S',t'$, $p$ satisfies  $\widehat{d}(p,S') \leq 10\alpha t + \widehat{d}(S,S') + \alpha t'$.
\State \Return $Y$

\end{algorithmic}

\end{algorithm}

Now we proceed with the proof that the algorithm computes a hierarchical partition with the properties that we need. By the core-outlier \Cref{lemma:coreoutliereff}, we have the following statement. 
\begin{lemma}
Every seed set $S$ at scale $t$ which forms a clique in $\widehat{G}_t$ contains a core of $80 \log n$ vertices whose true diameter is at most $\alpha t$, where $\alpha = 2$. The rest of the at most $20\log n$ vertices will be called outliers.
\end{lemma}

For a set $S$ of size $100\log n$, we fix any core of $S$ as defined above and call it $\mathrm{core}(S)$.

\begin{lemma}[Distance between cores lemma]\label{lemma:distancecores}
With high probability over the corruption, the following holds simultaneously.
Let \(S,S'\) be any two disjoint sets of size \(100\log n\), and let
\(A\subseteq S\), \(A'\subseteq S'\) satisfy $|A|,|A'|\ge 80\log n$. If
$\widehat d(S,S')\le x$,
then there exist points \(a\in A\) and \(a'\in A'\) such that $d(a,a')\le x$.

In particular, this applies when \(A=\mathrm{core}(S)\) and
\(A'=\mathrm{core}(S')\).
\end{lemma}
\begin{proof}
Let \(L=\log_2 n\).  We first prove a high-probability event that holds
simultaneously for all pairs of large disjoint subsets.

Fix two disjoint sets \(A,A'\subseteq X\) with $80L\le |A|,|A'|\le 100L$.
Let $N:=|A||A'|$. Then $6400L^2\le N\le 10000L^2$. Let \(Z_{A,A'}\) be the number of corrupted pairs between \(A\) and \(A'\).
By independence, \(Z_{A,A'}\) is stochastically dominated by $B\sim \mathrm{Bin}(N,0.1)$.
Let $\mu:=\EE[B]=0.1N$. We will bound the probability that $Z_{A,A'}\ge N-5000L^2$. First observe that
\[
\frac{N-5000L^2}{\mu}
=
\frac{N-5000L^2}{0.1N}
=
10-\frac{50000L^2}{N}.
\]
Since \(N\ge 6400L^2\), this ratio is at least $10-\frac{50000L^2}{6400L^2}
=
\frac{35}{16}$.
Thus $N-5000L^2\ge \frac{35}{16}\mu$. 
Equivalently, with \(\delta=19/16\), $N-5000L^2\ge (1+\delta)\mu$.
By the Chernoff bound,
\[
\Pr\!\left[Z_{A,A'}\ge N-5000L^2\right]
\le
\Pr[B\ge (1+\delta)\mu]
\le
\exp\left(-\frac{\delta^2}{2+\delta}\mu\right).
\]
Note that $\frac{\delta^2}{2+\delta}
=
\frac{(19/16)^2}{2+19/16}
=
\frac{361}{816}$, 
and $\mu=0.1N\ge 0.1\cdot 6400L^2=640L^2$.
Therefore
\[
\Pr\!\left[Z_{A,A'}\ge N-5000L^2\right]
\le
\exp\left(-\frac{361}{816}\cdot 640L^2\right)
\le
\exp(-280L^2).
\]

We now union bound over all choices of \(A,A'\) with sizes between
\(80L\) and \(100L\).  The number of ordered pairs \((A,A')\) of such sets is
at most
\[
\left(\sum_{r=80L}^{100L}\binom{n}{r}\right)^2
\le
(21L)^2 n^{200L}.
\]
Since \(L=\log_2 n\), this is $(21L)^2 2^{200L^2}
=
\exp\left(200\ln(2)L^2+2\ln(21L)\right)$. 
For every \(L\ge 1\),
$2\ln(21L)\le 7L^2$,
and also \(200\ln(2)<139\). Hence the number of choices is at most $\exp(146L^2)$.
Thus, with probability at least $1-\exp(146L^2)\exp(-280L^2)
=
1-\exp(-134L^2)$,
the following event holds simultaneously for all disjoint
\(A,A'\subseteq X\) with \(80L\le |A|,|A'|\le 100L\):
\[
Z_{A,A'}< |A||A'|-5000L^2.
\]

Condition on this event.  Let \(S,S'\subseteq X\) be sets of size \(100L\), and
let \(A\subseteq S\), \(A'\subseteq S'\) satisfy
$|A|,|A'|\ge 80L$.

Let $N:=|A||A'|$.
Suppose, for contradiction, that $d(a,a')>x
$, for every $a\in A,\ a'\in A'$.
Since \(\widehat d(S,S')\le x\), by definition there are at least
\(5000L^2\) edges between \(S\) and \(S'\) whose observed distance is at most
\(x\).  Among all possible such edges, at most
$|S||S'|-|A||A'|
=
10000L^2-N$
do not have one endpoint in \(A\) and the other endpoint in \(A'\).  Therefore at least $5000L^2-(10000L^2-N) = N-5000L^2$
edges with one endpoint in \(A\) and the other endpoint in \(A'\) have observed
distance at most \(x\).

By the contradiction assumption, every edge between \(A\) and \(A'\) has true
distance larger than \(x\).  Thus every edge between \(A\) and \(A'\) whose
observed distance is at most \(x\) must be corrupted.  Hence there are at least $N-5000L^2$
corrupted pairs between \(A\) and \(A'\), contradicting the high-probability
event.

Therefore there exist \(a\in A\) and \(a'\in A'\) such that $d(a,a')\le x$. This proves the lemma.
\end{proof}

\begin{lemma}[Chaining seeds]\label{lemma:chainingseeds}
Let \(S_1,S_2,\ldots,S_k\) be disjoint seed sets appearing in this
order along an ancestor chain, with scales
\(t_1 < t_2 < \cdots < t_k\) (that is, $S_1$ is a seed at scale $t_1$, $S_2$ is a seed at scale $t_2 > t_1$ and so on, and the corresponding clusters all lie on a leaf-to-root path in the hierarchical tree). Suppose $\widehat{d}(S_i, S_{i+1}) \leq u_i$ for all $i \in [k-1]$. Then $\widehat{d}(S_1, S_k) \leq \sum_{i \in [k-1]} u_i + 4\alpha t_k$.
\end{lemma}

\begin{proof}
The proof follows from the previous lemma, together with the observation that the cores at scale $t$ have small diameter (at most $\alpha t$). Apply the previous lemma to $S_i$ and $S_{i+1}$ - it follows that there exist two points $q_i,q_{i+1} \in \mathrm{core}(S_i), \mathrm{core}(S_{i+1})$ such that $d(q_i, q_{i+1}) \leq u_i$. Applying this for each $i$, and observing that $\mathrm{core}(S_i)$ has diameter $\alpha t_i$, we obtain points $q_1, q_k$ in the cores of $S_1,S_k$ which are at true distance at most $\sum_{i \in [k-1]} u_i + 2\alpha \sum_{i \neq 1,k} t_i$.

Now it follows that any two points in the cores of $S_1,S_k$ are at most $d: = \sum_{i \in [k-1]} u_i + 2\alpha \sum_{i \in [k]} t_i$ apart in true distance, since the diameters of these cores are at most $\alpha t_1$ and $\alpha t_k$ respectively.

This means there are at least $6400\log^2 n$ pairs of points in $S_1 \times S_k$ which are at true distance at most $d$, which in turn means that $\widehat{d}(S_1, S_k) \leq d$. Since scales are powers of $2$, we must have $t_1 + t_2 + \cdots + t_k \leq 2t_k$. This finishes the proof.   
\end{proof}
\begin{lemma}[Point-seed soundness]\label{lemma:pointseedsound}
There is an absolute constant \(c_{\mathrm{out}}\) such that, with high
probability over the corruption, the following holds simultaneously.  Let
\(S\) be a set of \(100\log n\) points, let \(A\subseteq S\) satisfy
$|A|\ge 80\log n$,
and let \(R\ge 0\). 
Then there are at most \(c_{\mathrm{out}}\log n\) points
\(p\in X\) such that $\widehat d(p,S)\le R$
but for every $a\in A$, $d(p,a)>R$.
\end{lemma}

\begin{proof}
Let \(L=\log_2 n\).  We prove the statement with \(c_{\mathrm{out}}=11\).
Set $s:=\lceil 10L\rceil$.
We first prove a high-probability event.  Fix a set \(A\subseteq X\) with
$80L\le |A|\le 100L$ and a set \(P\subseteq X\setminus A\) with \(|P|=s\).  Let \(Z_{P,A}\) be the
number of corrupted pairs with one endpoint in \(P\) and the other endpoint in
\(A\). 
Let $N:=|P||A|=s|A|$. Then \(Z_{P,A}\) is stochastically dominated by
$B\sim \mathrm{Bin}(N,0.1)$. Let $\mu:=\EE[B]=0.1N$.
Since \(|A|\le 100L\), we have $\mu=0.1s|A|\le 10Ls$. Therefore $60Ls\ge 6\mu$.
By the Chernoff bound,
\[
\Pr[Z_{P,A}\ge 60Ls]
\le
\Pr[B\ge 60Ls]
\le
\Pr[B\ge 6\mu]
\le
\exp\left(-\frac{25}{7}\mu\right).
\]
Also, since \(|A|\ge 80L\), $\mu=0.1s|A|\ge 8Ls$.
Thus
\[
\Pr[Z_{P,A}\ge 60Ls]
\le
\exp\left(-\frac{25}{7}\cdot 8Ls\right)
=
\exp\left(-\frac{200}{7}Ls\right).
\]
Since \(s=\lceil 10L\rceil\ge 10L\), this gives $\Pr[Z_{P,A}\ge 60Ls]
\le
\exp\left(-\frac{2000}{7}L^2\right)$.
Now union bound over all choices of \(A\) and \(P\).  The number of possible
sets \(A\) is at most
$\sum_{r=80L}^{100L}\binom{n}{r}
\le
21L\cdot n^{100L}$.
The number of possible sets \(P\) of size \(s\) is at most $\binom{n}{s}\le n^s$.

Since \(L\ge 1\), we have \(s=\lceil 10L\rceil\le 11L\).  Hence the total
number of pairs \((A,P)\) is at most
$21L\cdot n^{100L+s}
\le
21L\cdot n^{111L}$.
Using \(L=\log_2 n\), this is
$21L\cdot 2^{111L^2}
=
\exp\left(111\ln(2)L^2+\ln(21L)\right)$.
For every \(L\ge 1\), \(\ln(21L)\le 4L^2\), and \(111\ln(2)<77\).  Therefore
the number of choices is at most $\exp(81L^2).$
Thus, with probability at least $1-\exp(81L^2)\exp\left(-\frac{2000}{7}L^2\right)
\ge
1-\exp(-200L^2)$,
the following event holds simultaneously for every such \(A\) and \(P\): $Z_{P,A}<60Ls$.

Condition on this event.  Let \(S\) be a set of \(100L\) points, let
\(A\subseteq S\) have size at least \(80L\), and let \(R\ge 0\).  Suppose for
contradiction that there are at least \(s=\lceil 10L\rceil\) points \(p\in X\)
such that $\widehat d(p,S)\le R$ but for every $a\in A$, $d(p,a)>R$.
Choose any set \(P\) of exactly \(s\) such points.
First note that \(P\cap A=\emptyset\).  Indeed, if \(p\in A\), then taking
\(a=p\) gives \(d(p,a)=0\le R\), contradicting the condition
\(d(p,a)>R\) for every \(a\in A\).
Now fix \(p\in P\).  Since \(\widehat d(p,S)\le R\), there are at least \(80L\)
points \(s'\in S\) such that the observed distance between \(p\) and \(s'\) is
at most \(R\).  At most
$|S\setminus A|\le 20L$ of these points lie outside \(A\). 
Hence there are at least $80L-20L=60L$
points \(a\in A\) whose observed distance from \(p\) is at most \(R\).  But by assumption \(d(p,a)>R\) for every \(a\in A\).  Therefore each of these at least
\(60L\) pairs \((p,a)\) must be corrupted.

Since this holds for every \(p\in P\), there are at least $60L\cdot |P|=60Ls$
corrupted pairs between \(P\) and \(A\).  This contradicts the high-probability event, which says that every such pair \(P,A\) has fewer than \(60Ls\) corrupted
pairs between them.
Therefore there are fewer than \(\lceil 10L\rceil\) bad points.  In particular,
for \(L\ge 1\), there are at most \(11L\) bad points. Thus the lemma holds
with \(c_{\mathrm{out}}=11\).
\end{proof}

\begin{lemma}\label{lemma:diameterhier}
For every cluster $C$ formed at scale $t'$, either $C$ has size at most $3c_{out} \log n$ or there is a core of $C$, obtained by removing at most $c_{out}\log n$ points from $C$, that has true diameter of at most $\beta t'$ for some constant $\beta = O(\alpha)$.

\end{lemma}

\begin{proof}
Let $S$ be the seed of cluster $C$. Consider $p \in C$. Let $S_p$ be the seed at the bottommost level in which $p$ was in a non-singleton cluster. Further, let the scale at this level be $t$.

By the algorithm description, since $p$ was accepted in the cluster of $S_p$, we must have $\widehat{d}(p,S) \leq 10\alpha t + \widehat{d}(S_p, S) + \alpha t'$ by the execution of the algorithm.

To bound $\widehat{d}(S_p, S)$, we look at the chain of seed vertices formed between $S_p$ and $S$, each going one level up. Let $S_p = S_1, S_2, \ldots, S_k = S$ be this chain. We assume that every $S_i$ is distinct: the only other case is that we may have indices $i$ for which $S_i = S_{i-1}$, however we can simply omit those indices and the rest of the proof follows verbatim. Let the corresponding scales be $t = t_1, \ldots, t_k = t'$. By the construction of the algorithm, we know that $\widehat{d}(S_i, S_{i +1})$ is at most $10 \alpha t_{i+1}$. Apply the chaining seeds lemma,~\Cref{lemma:chainingseeds} to $S_1,S_2 \ldots S_k$ with $u_i = 10\alpha t_{i+1}$. We obtain that $\widehat{d}(S_1, S_k) \leq 10\alpha(t_1 + t_2 + \cdots + t_k) + 4\alpha t_k \leq 30 \alpha t_k = 30\alpha t'$.  

Therefore we must have $\widehat{d}(p,S) \leq 10 \alpha t +  \widehat{d}(S_p, S) + \alpha t' \leq 10\alpha t + 30\alpha t' + \alpha t' \leq 50 \alpha t'$.

We have proved that for every $p \in C$, $\widehat{d}(p,S) \leq \alpha' t'$ where $\alpha' = 50 \alpha$. By Lemma~\ref{lemma:pointseedsound}, applied with
\(A=\mathrm{core}(S)\) and \(R=\alpha' t'\), all but
\(c_{\mathrm{out}}\log n\) points \(p\in C\) have true distance at most
\(\alpha't'\) from some point of \(\mathrm{core}(S)\). Finally, \(\mathrm{core}(S)\) itself has diameter at most \(\alpha t'\).
Thus, after removing the exceptional \(c_{\mathrm{out}}\log n\) points, any two
remaining points \(p,q\in C\) satisfy
\[
d(p,q)\le \alpha't' + \alpha t' + \alpha't'
= (2\alpha' + \alpha)t'.
\]
By taking \(\beta=200\alpha\), the remaining set
has true diameter at most \(\beta t'\).

Thus either $C$ has size at most $3c_{out} \log n$ or one can remove $c_{out}\log n$ points from $C$ to obtain a set with true diameter at most $\beta t'$ where $\beta = 200\alpha$.
\end{proof}

Now that we have proved the property about diameter, it remains to argue about the cost of the hierarchical partitioning. The following lemma will be the key to lower bound the cost of the true MST.

\begin{lemma}\label{lemma:sep}
Suppose there is a set $W$ of $m$ points in the metric space $(X,d)$ such that for any pair among these points, the true distance is at least $t$. Then the cost of the MST is at least $\frac{(m - 1)t}{2}$.
\end{lemma}

\begin{proof}
Consider any spanning tree $T$. Take an Euler tour of $T$, and shortcut the Euler tour by viewing it as paths between vertices of $W$. Each distance along such a path is at least $t$, and there are at least $(m - 1)$ such paths (in fact, except for the degenerate case $m = 1$, there are $m$ such paths). Thus the total length of the Euler tour is at least $(m-1)t$, and hence the cost of $T$ must be at least $\frac{(m-1)t}{2}$. 
\end{proof}

\begin{lemma}
Let $C_1, C_2$ be two distinct clusters produced by Algorithm~\ref{alg:hierarchical-partition} at scale $t$, with seed sets $S_1, S_2$. Then the true distance between any two points in the cores of $S_1$ and $S_2$ is at least $t$. 
\end{lemma}

\begin{proof}
Suppose this is not the case, so that there exist two points in the cores of $S_1,S_2$ that are at distance at most $t$. Then $\mathrm{core}(S_1) \cup \mathrm{core}(S_2)$ has a diameter of at most $t + 2\alpha t \leq 10\alpha t$. Note that this implies $\widehat{d}(S_1, S_2) \leq 10 \alpha t$.

Assume without loss of generality that $S_1$ was considered as a seed before $S_2$ by the algorithm. There are two cases. The first is that $S_2$ was already a seed from the previous level (so that it is an object in $\mathcal{L}$). If so, since $\widehat{d}(S_1, S_2) \leq 10\alpha t$, $S_2$ should have been carved while considering $S_1$, a contradiction.

The second case is that $S_2$ is a ``newly germinated'' seed at this level by the coming together of $100\log n$ singletons. Pick any $s_2 \in \mathrm{core}(S_2)$. We have $\widehat{d}(s_2, S_1) \leq 10 \alpha t$. Second, for any other seed $S'$ at scale $t'$ such that $\widehat{d}(S_1, S') = x$, we know by \Cref{lemma:distancecores} that there exist points $q_1 \in \mathrm{core}(S_1), q_2 \in \mathrm{core}(S')$ such that the true distance $d(q_1, q_2) \leq x$. Further since $s_2 \in \mathrm{core}(S_2)$ and $q_1 \in \mathrm{core}(S_1)$ and $\mathrm{core}(S_1) \cup \mathrm{core}(S_2)$ has diameter $10\alpha t$, we have $d(s_2, q_1) \leq 10\alpha t$. But this means that the distance of $s_2$ from any point in $\mathrm{core}(S')$ is at most $10\alpha t + x + \alpha t'$. This in turn would imply $\widehat{d}(s_2, S') \leq 10\alpha t + x + \alpha t'$, and hence $s_2$ would have been carved while considering $S_1$ as a seed, a contradiction.
\end{proof}

\begin{lemma}
Let $k_1$ be the number of non-singleton clusters at scale $t$. Then the cost of MST is at least $\frac{(k_1-1)t}{2}$.

\end{lemma}

\begin{proof}
Choose one point from the core of the seed of each non-singleton cluster at
scale \(t\).  By the previous separation lemma, these \(k_1\) points are
pairwise at true distance at least \(t\).  Applying Lemma~\ref{lemma:sep} with
\(m=k_1\) yields
$
\mathrm{MST}\ge \frac{(k_1-1)t}{2}$.
\end{proof}

\begin{lemma}
Let $k_2$ be the number of singleton clusters at scale $t$. Then the cost of MST is at least $\frac{1}{4}\left(\frac{k_2}{100\log n} - 1\right)t$.

\end{lemma}

\begin{proof}
Consider the situation when the algorithm at scale $t$ stopped, and suppose that there exists some set $S$ of remaining $100\log n$ singleton points that form a clique  in $G_t$. We show that this cannot happen. This is because, first, they would form a clique in $\widehat{G}_t$ as well. Second, $S$ would be admissible. Fix any set $S'$ which forms a clique at scale $t' \geq t$, and suppose $\widehat{d}(S,S') = x$. By the distance between cores lemma~\Cref{lemma:distancecores}, we know that there exist two points $q_1,q_2$ in the cores of $S$ and $S'$ at true distance at most $x$. $S$ itself is a clique hence its diameter is at most $t$, and the diameter of the core of $S'$ is at most $\alpha t'$. Applying the triangle inequality, for any $p \in S$, we must have $\widehat{d}(p,S') \leq t + x + \alpha t' < 10\alpha t + \widehat{d}(S, S') + \alpha t'$, and hence, admissibility follows.

Now consider any such singleton point $u$. Let $B\left(u,\frac{t}{2}\right)$ denote the ball of radius $\frac{t}{2}$ around $u$ in the true metric restricted to only the remaining singleton points. Then this entire ball forms a clique in $G_t$, and hence $|B\left(u,\frac{t}{2}\right)| \leq 100 \log n$. Repeatedly removing such balls, we obtain a subset of singletons of size at least $\frac{k_2}{100 \log n}$ such that every pair is at distance at least $\frac{t}{2}$. Applying~\Cref{lemma:sep}, the proof is complete. 
\end{proof}

\begin{lemma}
The cost of the hierarchical partition is bounded by $O(MST \log n \log \Delta')$.
\end{lemma}

\begin{proof}
If there are $n_t$ sets at level $t = 2^i$, then the cost of the hierarchical partitioning is at most $tn_t$. From the previous two lemmas, $tn_t$ is at most $O(\log n)(MST + t)$. Adding across all scales $t$, and observing that $MST \geq \Delta$, we obtain the lemma.
\end{proof}

Finally, our approximation guarantee of $O(\log^2 n\log^2 \Delta')$ now follows by combining with~\Cref{alg:rounding}.

\subsection{Removing the assumptions on distances and reducing the aspect ratio.}
Now we show how to remove the assumption that the algorithm knows $\Delta'$, and obtain a $\mathrm{polylog}(n)$ approximation.

First we compute $U$, the smallest value $t$  such that $\widehat{G}_{t}$ is fully a clique. Equivalently, $U = \max_{v_1,v_2 \in X} \widehat{d}(v_1,v_2)$. Since we are in the decrease-only setting, clearly $U \leq \Delta$. But we can show $U \geq \frac{\Delta}{2}$ with high probability as follows. 

Let $(a,b)$ be a pair of points in $X$ that are at true distance $\Delta$ apart. Then for any other point $c$, either the pair of vertices $(a,c)$ are at least $\frac{\Delta}{2}$ apart or the pair of vertices $(b,c)$ are at least $\frac{\Delta}{2}$ apart. Thus to get $U < \frac{\Delta}{2}$, at least $n - 2$ edges incident to $a$ and $b$ must be corrupted. The probability of this is at most $p^{n-2}$, so for any $p < 1$, this does not happen with high probability. 

Next, we will reduce the aspect ratio to $n^2$ as follows. Let $\rho = \frac{2U}{n^2}$. From the given metric $d$, consider the metric $d_{\rho}$ by setting $d_{\rho}(u,v) = \max(d(u,v), \rho)$. The corresponding corrupted distances $\widehat{d}_{\rho}(u,v)$ are given as $\widehat{d}_{\rho}(u,v)= \max(\widehat{d}(u,v), \rho)$. Now we are guaranteed that the minimum distance is at least $\frac{2U}{n^2}$ and the maximum distance is at most $2U$. The value of the true MST increases by at most $\frac{2U}{n^2} \cdot (n - 1) \leq O(\MST(X)/n)$, so this preserves any approximation. Finally, scale each distance by $\frac{1}{\rho}$. The minimum distance is now at least $1$, and the maximum distance is at most $n^2$. We can now set $\Delta' = n^2$ to use our algorithm and obtain an $O(\mathrm{\polylog}(n))$ approximation, which we can again scale back by $\rho$.

Finally, we note that all high-probability events used in this section fail with probability $n^{-\omega(1)}$. Since any spanning tree has cost at most $(n-1)\Delta\leq (n-1)\MST(X)$, their failure contributes only $o(\MST(X))$ to the expected cost.

\section{Lower Bounds.}
\label{sec:lower}
In this section, we prove impossibility results, proving that some of our assumptions on parameters were necessary. 

\paragraph{General model.}
First, we show that the assumption $p < 1/2$ is necessary in the general model.%

\generalhardness*

\begin{proof}
Let $a = \lfloor n/2 \rfloor$ and $b = n-a$. 
We first describe a distribution over true metrics. Choose a uniformly random partition
\[
    X = A \sqcup B, \qquad |A|=a,\ |B|=b.
\]
Given this partition, define a metric $d=d_{A,B}$ by
\[
    d(u,v)=
    \begin{cases}
        0, & \text{if } u,v \text{ lie in the same part},\\
        1, & \text{if } u,v \text{ lie in different parts}.
    \end{cases}
\]

For this metric, an MST consists of a spanning tree inside $A$, a spanning tree inside $B$, and one edge crossing between $A$ and $B$. Hence
$
    \MST(X)=1.
$

We now describe the adversary.
When a pair is uncorrupted, its observed value is its true distance. When a pair is corrupted, the adversary acts as follows. If $u,v$ lie in the same part, the adversary sets
$
    \dhat(u,v)=1.
$
If $u,v$ lie in different parts, the adversary sets
$
    \dhat(u,v)=0$
with probability
$
    q=(1-p)/p
$
and sets $\dhat(u,v)=1$ otherwise. Since $p\geq 1/2$, we have $q \in [0,1]$.

Now consider any fixed pair $\{u,v\}$. If $u,v$ lie in the same part, then
\[
    \Pr[\dhat(u,v)=0]=1-p,
\]
because the pair has observed value $0$ exactly when it is uncorrupted. If $u,v$ lie in different parts, then
\[
    \Pr[\dhat(u,v)=0]=p \cdot q
    = p \cdot \frac{1-p}{p}
    = 1-p.
\]
Thus, regardless of whether $u,v$ lie in the same part or in different parts,
\[
    \Pr[\dhat(u,v)=0]=1-p,
    \qquad
    \Pr[\dhat(u,v)=1]=p.
\]
Moreover, these observations are independent across pairs. Therefore the distribution of the entire corrupted distance matrix $\dhat$ is independent of the hidden partition $A \sqcup B$.

It follows that the output tree $T=\mathcal A(\dhat)$ is independent of the hidden partition. Now fix any tree $T$ on $X$. The expected number of crossing edges of $T$ is at least $n/4$.
Hence, for every fixed tree $T$,
\[
    \EE_{A,B}[\cost(T)] \geq \frac{n}{4}.
\]
Since the random tree $\mathcal A(\dhat)$ is independent of the hidden partition, the same bound holds after averaging over the corrupted distance matrix and over the internal randomness of $\mathcal A$:
\[
    \EE_{A,B,\dhat,\mathcal A}\big[\cost(\mathcal A(\dhat))\big]
    \geq \frac{n}{4}.
\]

Therefore, by averaging, there exists a fixed partition $A \sqcup B$ such that the corresponding metric $d=d_{A,B}$ satisfies
\[
    \EE_{\dhat,\mathcal A}\big[\cost(\mathcal A(\dhat))\big]
    \geq \frac{n}{4}
\]
under the adversary above. Since $\MST(X) = 1$ for this metric, we get
\[
    \EE_{\dhat,\mathcal A}\big[\cost(\mathcal A(\dhat))\big]
    \geq \frac{n}{4}\cdot \MST(X).
\]
\end{proof}

\paragraph{Decrease-only.}

For the decrease-only model, 
we relate our problem to {\em refuting $q$-colorability of $G(n,1/2)$}. 
It is well known that $k = \Theta(\sqrt{n})$ is the threshold for the existence of a polynomial-time algorithm refuting that $G(n, 1/2)$ does not have a $k$-clique; more formally, when $k \geq n^{1/2 + \eps}$ for some constant $\eps > 0$, there is a polynomial-time algorithm $P$ such that $P(G) = 1$ whenever $G$ has a clique of size $k$ and $P(G) = 0$ with high probability when $G \sim G(n, 1/2)$, and if $k \leq n^{1/2 - \eps}$, such an algorithm is unlikely to exist, conjectured as the Planted Clique Hypothesis~\cite{Jerrum1992LargeCliques, Kucera1995ExpectedComplexity, AlonKrivelevichSudakov1998HiddenClique} and evidenced by the strong Sum-of-Squares (SoS) lower bound~\cite{barak2019nearly}. 

How about refuting $q$-colorability of $G(n, 1/2)$?
The chromatic number of $G(n, 1/2)$ is $\Theta(n/\log n)$ with high probability~\cite{bollobas1988chromatic}, and it is natural to suspect that $q = n/k = \sqrt{n}$ is again the threshold of polynomial-time refutability.
While this remains an interesting open problem, Kothari, Vempala, Wein, and Xu~\cite{kothari2023planted}, for $q = n^{2/3+\eps}$ for any constant $\eps > 0$, proposed the following {\em quietly planted} distribution $P_q(n)$ for $q$-colorable graphs. 

\begin{definition} [Definition 2.18 of~\cite{kothari2023planted}] Suppose $n, q$ are positive integers. To each of the $n$ vertices, independently assign a label $(a, b) \in [q] \times [q]$ uniformly at random. Conditioned on the labels, do the following independently for each pair of distinct vertices $\{u, v \}$: denote the two vertex labels by $(a_1, b_1)$ and $(a_2, b_2)$; if $a_1 = a_2$ then do not include the edge $(u, v)$; if $a_1 \neq a_2$ and $b_1 = b_2$ then include the edge $(u, v)$; otherwise include the edge $(u, v)$ with probability $1/2$. 
\end{definition}

Then they gave partial evidence that no polynomial-time algorithm will (weakly) separate $G(n,1/2)$ and $P_q(n)$ by ruling out algorithms based on degree-$o(\log n / \log \log n)^2$ polynomials. We refer the reader to \cite{kothari2023planted} for their exact definition of weak separation, but it includes a decision (i.e., outputs {\em accept} or {\em reject}) algorithm $A$ that accepts $G \sim P(n,q)$ with probability at least $0.7$ and accepts $G \sim G(n,1/2)$ with probability at most $0.3$. 
Such partial evidence based on low-degree polynomials, motivated by the fact that $O(\log n)$-degree polynomials capture the power of many polynomial-time algorithms for average-case tasks, has been the focus of active research~\cite{hopkins2018statistical, buhai2025quasi, wein2025computationalcomplexitystatisticsnew}. 

We prove the following lemma, showing that a $\widetilde{O}(n/q)$-approximation algorithm for our problem can be used to separate these two distributions, which suggests that achieving such an approximation likely requires superpolynomial time. 

\begin{theorem}
Let $p = 1/2$, $q = n^{2/3 + \eps}$ for any constant $\eps > 0$, and suppose there exists an $n/(q \log^3 n)$-approximation for the MST in the decrease-only model. Then there is an algorithm $\mathcal A$ that accepts $G \sim P(n, q)$ with probability $1-O(1/\log n)$ and rejects $G \sim G(n, 1/2)$ with high probability.  
\label{thm:hard-decrease}
\end{theorem}
\begin{proof}
Given an $n/(q \log^3 n)$ approximation for the MST, the algorithm $\mathcal A$, given a graph $G = (V, E)$, does the following: 

\begin{enumerate}
    \item Construct $\hat d$ such that $\hat d(u, v) = 1$ if $(u, v) \in E$ and $0$ otherwise. 
    \item Run the $n/(q \log^3 n)$ approximation algorithm on $\hat d$ to obtain a spanning tree $T$ of $K_n$.
    \item Let $S = (v_1, \dots, v_n)$ be the preorder traversal of $T$. 
    \item Perform the following {\em greedy coloring} of $G$:  compute the longest prefix of $S$ that forms an independent set, color it using a new color, and remove it from $S$, until $S$ becomes empty. 
    \item If the number of colors used is at most $n/\log^2 n$, accept. Otherwise, reject. 
\end{enumerate}
The fact that $G \sim G(n,1/2)$ is rejected with high probability is clear, because the chromatic number of $G \sim G(n, 1/2)$ is $\Omega(n/\log n)$ with high probability~\cite{bollobas1988chromatic}. 

Now we claim that $G \sim P(n, q)$ will be accepted with probability $1-O(1/\log n)$. One can view $\hat d$ as the corrupted metric generated by the following procedure. From $P(n, q)$, after sampling label $(a, b)$ for every vertex, the true metric $d$ is a $q$-point metric where two vertices have distance $0$ if their $a$-label is equal and $1$ otherwise. Then, each pair of vertices is corrupted with probability $1/2$ independently, and the adversary lowers the distance to $0$ if and only if their $b$-labels are different. 

Since $d$ has the MST cost $q-1$, the approximation algorithm returns a tree $T$ whose expected cost in $d$ is $O(q \cdot n/(q \log^3 n)) = O(n/\log^3 n)$, so, by Markov's inequality, with probability $1-O(1/\log n)$, the cost of $T$ in $d$, which is the number of edges of $T$ connecting vertices with different $a$-values (or planted colors), is at most $n/(2 \log^2 n)$. Then, in the preorder traversal $S$, there are at most $n / \log^2 n$ pairs of consecutive vertices with different $a$-values, so the greedy coloring will color $G$ with at most $n / \log^2 n$ colors. 
\end{proof}

\appendix
\section{Hardness of Computing a Minimum Dominating Tree.}
\label{app:dominating-tree-hardness}

In this appendix, we show that computing a minimum dominating tree from the
observed distances is at least as hard as computing the chromatic number of a
graph. The hardness holds even when every observed distance belongs to
$\{0,1\}$.

Let $X$ be a finite point set and let
$\widehat d \colon \binom{X}{2} \to \{0,1\}$ be a symmetric observed-distance
function. We do not require $\widehat d$ to be a metric. For a spanning tree
$T$ on $X$, define its observed cost by
\[
    \widehat c(T)
    :=
    \sum_{\{u,v\}\in E(T)} \widehat d(u,v),
\]
and let $\widehat d_T(u,v)$ denote the sum of the observed edge lengths along
the unique $u$--$v$ path in $T$. Recall that $T$ is dominating if
\[
    \widehat d_T(u,v) \geq \widehat d(u,v)
    \qquad
    \text{for every distinct }u,v\in X.
\]

\begin{theorem}[Graph-coloring hardness]%
\label{thm:dominating-tree-hardness}%
Computing a minimum-cost dominating tree is NP-hard, even when all observed
distances belong to $\{0,1\}$.

More precisely, given a graph $G=(V,E)$, define the observed distances
$\widehat d_G$ on $V$ by
\[
    \widehat d_G(u,v)
    :=
    \begin{cases}
        1, & \text{if }\{u,v\}\in E,\\
        0, & \text{if }\{u,v\}\notin E.
    \end{cases}
\]
Then
\[
    \min_{\substack{T\text{ spanning tree}\\T\text{ dominating}}}
        \widehat c(T)
    =
    \chi(G)-1,
\]
where $\chi(G)$ is the chromatic number of $G$.

Consequently, for every integer $k\geq 1$, the graph $G$ is $k$-colorable if
and only if there exists a dominating tree of observed cost at most $k-1$.
In particular, the corresponding decision problem is NP-complete, and inherits the approximation hardness of graph coloring.
\end{theorem}

\begin{proof}
We first show how a dominating tree gives a proper coloring. Let $T$ be any
dominating tree, and write
\[
    r := \widehat c(T).
\]
Because every edge of $T$ has observed length either $0$ or $1$, the value
$r$ is exactly the number of unit-length edges of $T$. Delete these
unit-length edges. Since $T$ is a tree, the remaining zero-length forest has
exactly $r+1$ connected components.

We claim that each such component is an independent set in $G$. Indeed, let
$u$ and $v$ belong to the same component. The entire $u$--$v$ path in $T$
then consists of zero-length edges, and hence
\[
    \widehat d_T(u,v)=0.
\]
Since $T$ is dominating,
\[
    \widehat d_G(u,v)
    \leq
    \widehat d_T(u,v)
    =
    0.
\]
It follows that $\widehat d_G(u,v)=0$, and therefore
$\{u,v\}\notin E$. Thus, assigning one color to each zero-length component
gives a proper coloring of $G$ with $r+1$ colors. Consequently,
\[
    \chi(G)\leq r+1,
\]
and hence every dominating tree satisfies
\[
    \widehat c(T)\geq \chi(G)-1.
\]

For the reverse inequality, let
\[
    V=C_1\mathbin{\dot\cup}C_2\mathbin{\dot\cup}\cdots
      \mathbin{\dot\cup}C_q,
    \qquad q=\chi(G),
\]
be an optimal proper coloring of $G$. Every pair of vertices in the same
color class is a non-edge of $G$, and therefore has observed distance zero.
For each $i$, connect the vertices of $C_i$ by an arbitrary spanning tree
$F_i$. Every edge of $F_i$ has observed length zero.

Moreover, any two distinct color classes $C_i$ and $C_j$ have at least one
edge of $G$ between them. Otherwise, $C_i\cup C_j$ would be an independent
set, and the two classes could be merged, contradicting the optimality of
the coloring. Thus, for every $i\in\{2,\ldots,q\}$, we may choose an edge
\[
    e_i\in E
\]
with one endpoint in $C_1$ and the other endpoint in $C_i$.

Let
\[
    T
    :=
    \left(\bigcup_{i=1}^q F_i\right)
    \cup
    \{e_2,\ldots,e_q\}.
\]
After contracting each $F_i$, the selected edges form a star centered at
$C_1$. Therefore $T$ is a spanning tree. Its only unit-length edges are
$e_2,\ldots,e_q$, so
\[
    \widehat c(T)=q-1=\chi(G)-1.
\]

It remains to verify that $T$ is dominating. If
$\widehat d_G(u,v)=0$, then the domination inequality is immediate because
all tree distances are nonnegative. If $\widehat d_G(u,v)=1$, then
$\{u,v\}\in E$, so $u$ and $v$ lie in different color classes. Their path
in $T$ must therefore cross at least one of the unit-length edges
$e_2,\ldots,e_q$. Hence
\[
    \widehat d_T(u,v)\geq 1=\widehat d_G(u,v).
\]
Thus $T$ is dominating, proving
\[
    \min_{\substack{T\text{ spanning tree}\\T\text{ dominating}}}
        \widehat c(T)
    \leq
    \chi(G)-1.
\]
Combining the two inequalities proves the claimed equality.

The reduction clearly takes polynomial time. Membership of the decision
problem in NP follows because, given a tree, one can compute its cost and
verify all domination inequalities in polynomial time.
\end{proof}

The proof also gives an explicit correspondence: deleting the unit-length
edges of any optimal dominating tree recovers an optimal coloring of the
original graph.

\section{AI Disclosure.}

GPT 5.4/5.5 Thinking/Pro was used to prove an earlier version of \Cref{lemma:dominatinglog} (the authors earlier proved an $O(\log^2 n)$ bound). We also used it to write proofs of elementary probability lemmas (\Cref{lemma:coreoutliereff,lemma:pointseedsound,lemma:distancecores}), adapt hardness results to the precise setting (given input from the authors), show the hardness of computing dominating trees (\Cref{thm:dominating-tree-hardness}) and to create figures and give feedback on parts of the paper. 

All the algorithmic approaches in this paper were discovered by the authors, and all AI-assisted proofs were verified by the authors.

\bibliographystyle{alpha}
\bibliography{ref}

@article{Jerrum1992LargeCliques,
  author  = {Jerrum, Mark},
  title   = {Large Cliques Elude the Metropolis Process},
  journal = {Random Structures \& Algorithms},
  volume  = {3},
  number  = {4},
  pages   = {347--359},
  year    = {1992},
  doi     = {10.1002/rsa.3240030402}
}

@article{Kucera1995ExpectedComplexity,
  author  = {Ku{\v{c}}era, Lud{\v{e}}k},
  title   = {Expected Complexity of Graph Partitioning Problems},
  journal = {Discrete Applied Mathematics},
  volume  = {57},
  number  = {2--3},
  pages   = {193--212},
  year    = {1995},
  doi     = {10.1016/0166-218X(94)00103-K}
}

@inproceedings{devlin2019bert,
  title={Bert: Pre-training of deep bidirectional transformers for language understanding},
  author={Devlin, Jacob and Chang, Ming-Wei and Lee, Kenton and Toutanova, Kristina},
  booktitle={Proceedings of the 2019 conference of the North American chapter of the association for computational linguistics: human language technologies, volume 1 (long and short papers)},
  pages={4171--4186},
  year={2019}
}

@inproceedings{he2016deep,
  title={Deep residual learning for image recognition},
  author={He, Kaiming and Zhang, Xiangyu and Ren, Shaoqing and Sun, Jian},
  booktitle={Proceedings of the IEEE conference on computer vision and pattern recognition},
  pages={770--778},
  year={2016}
}

@article{mikolov2013efficient,
  title={Efficient estimation of word representations in vector space},
  author={Mikolov, Tomas and Chen, Kai and Corrado, Greg and Dean, Jeffrey},
  journal={arXiv preprint arXiv:1301.3781},
  year={2013}
}

@article{AlonKrivelevichSudakov1998HiddenClique,
  author  = {Alon, Noga and Krivelevich, Michael and Sudakov, Benny},
  title   = {Finding a Large Hidden Clique in a Random Graph},
  journal = {Random Structures \& Algorithms},
  volume  = {13},
  number  = {3--4},
  pages   = {457--466},
  year    = {1998},
  doi     = {10.1002/(SICI)1098-2418(199810/12)13:3/4<457::AID-RSA14>3.0.CO;2-W}
}

@article{bollobas1988chromatic,
  title={The chromatic number of random graphs},
  author={Bollob{\'a}s, B{\'e}la},
  journal={Combinatorica},
  volume={8},
  number={1},
  pages={49--55},
  year={1988},
  publisher={Springer}
}

@article{buhai2025quasi,
  title={The quasi-polynomial low-degree conjecture is false},
  author={Buhai, Rares-Darius and Hsieh, Jun-Ting and Jain, Aayush and Kothari, Pravesh K},
  journal={arXiv preprint arXiv:2505.17360},
  year={2025}
}

@misc{wein2025computationalcomplexitystatisticsnew,
      title={Computational Complexity of Statistics: New Insights from Low-Degree Polynomials}, 
      author={Alexander S. Wein},
      year={2025},
      eprint={2506.10748},
      note={arXiv:2506.10748},
      archivePrefix={arXiv},
      primaryClass={math.ST},
      url={https://arxiv.org/abs/2506.10748}, 
}

@article{barak2019nearly,
  title={A nearly tight sum-of-squares lower bound for the planted clique problem},
  author={Barak, Boaz and Hopkins, Samuel and Kelner, Jonathan and Kothari, Pravesh K and Moitra, Ankur and Potechin, Aaron},
  journal={SIAM Journal on Computing},
  volume={48},
  number={2},
  pages={687--735},
  year={2019},
  publisher={SIAM}
}

@book{hopkins2018statistical,
  title={Statistical inference and the sum of squares method},
  author={Hopkins, Samuel},
  year={2018},
  publisher={Cornell University}
}

@inproceedings{kothari2023planted,
  title={Is planted coloring easier than planted clique?},
  author={Kothari, Pravesh and Vempala, Santosh S and Wein, Alexander S and Xu, Jeff},
  booktitle={The Thirty Sixth Annual Conference on Learning Theory},
  pages={5343--5372},
  year={2023},
  organization={PMLR}
}

@inproceedings{bateni2024metric,
  title={Metric clustering and mst with strong and weak distance oracles},
  author={Bateni, MohammadHossein and Dharangutte, Prathamesh and Jayaram, Rajesh and Wang, Chen},
  booktitle={The thirty seventh annual conference on learning theory},
  pages={498--550},
  year={2024},
  organization={PMLR}
}

@article{braverman2025relative,
  title={Relative Error Fair Clustering in the Weak-Strong Oracle Model},
  author={Braverman, Vladimir and Dharangutte, Prathamesh and Jiang, Shaofeng HC and Nguyen, Hoai An and Wang, Chen and Zhang, Yubo and Zhou, Samson},
  journal={Proceedings of Machine Learning Research},
  volume={267},
  pages={5418--5453},
  year={2025}
}

@article{pradhan2026improved,
  title={Improved Algorithms for Clustering with Noisy Distance Oracles},
  author={Pradhan, Pinki and Bhattacharya, Anup and Jaiswal, Ragesh},
  journal={arXiv preprint arXiv:2602.18389},
  year={2026}
}

@article{raychaudhury2026metric,
  title={Metric $k$-clustering using only Weak Comparison Oracles},
  author={Raychaudhury, Rahul and Esmailpour, Aryan and Galhotra, Sainyam and Sintos, Stavros},
  journal={arXiv preprint arXiv:2601.19333},
  year={2026}
}

@inproceedings{silwal2023kwikbucks,
  title={Kwikbucks: Correlation clustering with cheap-weak and expensive-strong signals},
  author={Silwal, Sandeep and Ahmadian, Sara and Nystrom, Andrew and McCallum, Andrew and Ramachandran, Deepak and Kazemi, Mehran},
  booktitle = {Proceedings of the 11th International Conference on Learning Representations (ICLR)},
  year      = {2023},
}

@article{blum1995coloring,
  title={Coloring random and semi-random k-colorable graphs},
  author={Blum, Avrim and Spencer, Joel},
  journal={Journal of Algorithms},
  volume={19},
  number={2},
  pages={204--234},
  year={1995},
  publisher={Elsevier}
}

@article{feige2001heuristics,
  title={Heuristics for semirandom graph problems},
  author={Feige, Uriel and Kilian, Joe},
  journal={Journal of Computer and System Sciences},
  volume={63},
  number={4},
  pages={639--671},
  year={2001},
  publisher={Elsevier}
}

@article{lykouris2021competitive,
  title={Competitive caching with machine learned advice},
  author={Lykouris, Thodoris and Vassilvitskii, Sergei},
  journal={Journal of the ACM (JACM)},
  volume={68},
  number={4},
  pages={1--25},
  year={2021},
  publisher={ACM New York, NY}
}

@inproceedings{kumar2018improving,
  title={Improving online algorithms via ML predictions},
  author={Kumar, Ravi and Purohit, Manish and Svitkina, Zoya},
  booktitle={Proceedings of the 32nd International Conference on Neural Information Processing Systems},
  pages={9684--9693},
  year={2018}
}

@article{mitzenmacher2020algorithms,
  title={Algorithms with Predictions},
  author={Mitzenmacher, Michael and Vassilvitskii, Sergei},
  journal={arXiv preprint arXiv:2006.09123},
  year={2020}
}

@article{diakonikolas2019robust,
  title={Robust estimators in high-dimensions without the computational intractability},
  author={Diakonikolas, Ilias and Kamath, Gautam and Kane, Daniel and Li, Jerry and Moitra, Ankur and Stewart, Alistair},
  journal={SIAM Journal on Computing},
  volume={48},
  number={2},
  pages={742--864},
  year={2019},
  publisher={SIAM}
}

@inproceedings{lai2016agnostic,
  title={Agnostic estimation of mean and covariance},
  author={Lai, Kevin A and Rao, Anup B and Vempala, Santosh},
  booktitle={2016 IEEE 57th Annual Symposium on Foundations of Computer Science (FOCS)},
  pages={665--674},
  year={2016},
  organization={IEEE}
}

@inproceedings{charikar2017learning,
  title={Learning from untrusted data},
  author={Charikar, Moses and Steinhardt, Jacob and Valiant, Gregory},
  booktitle={Proceedings of the 49th annual ACM SIGACT symposium on theory of computing},
  pages={47--60},
  year={2017}
}

@incollection{Feige20,
  author       = {Uriel Feige},
  editor       = {Tim Roughgarden},
  title        = {Introduction to Semirandom Models},
  booktitle    = {Beyond the Worst-Case Analysis of Algorithms},
  pages        = {189--211},
  publisher    = {Cambridge University Press},
  year         = {2020},
  url          = {https://doi.org/10.1017/9781108637435.013},
  doi          = {10.1017/9781108637435.013},
  bibsource    = {dblp computer science bibliography, https://dblp.org}
}
\end{document}